\RequirePackage{etex} 
\documentclass[a4paper,UKenglish,cleveref, autoref, thm-restate]{lipics-v2021}
\usepackage[utf8]{inputenc}
\usepackage{graphicx}
\usepackage{amsmath}
\usepackage{amssymb}
\usepackage{amsthm}
\usepackage{bbm}
\usepackage{cases}
\usepackage{hyperref, graphicx}
\usepackage{comment}
\usepackage{enumerate}
\usepackage{comment}
\usepackage[dvipsnames,table,xcdraw]{xcolor}
\usepackage{tabularx}

\usepackage{tikz, pgfplots}
\usetikzlibrary{positioning}
\usetikzlibrary{graphs,graphs.standard,quotes}
\usepackage{tikz-network}
\tikzstyle{vertex}=[draw,circle,color=black, color=black,inner sep=1.5pt]

\pgfplotsset{compat=1.18}

\usepackage{ulem}

\newcommand{\rate}{{grow pace}}
\newcommand{\sm}{\setminus}
\renewcommand{\emph}[1]{\textbf{#1}}

\newenvironment{proofclaim}[1][]%
	{\noindent {}{#1}{}}{ This proves Claim~(\arabic{claim}).\vspace{1.2ex}}

 \usepackage{soul} %Pour barrer en couleur 
\newcommand{\Cleo}[1]{{\color{purple} #1}}

\newcommand{\stc}[1]{\setstcolor{purple} \st{ #1}}
\title{How to Color Temporal Graphs to Ensure Proper Transitions}

\author{Allen {Ibiapina}}{Université Paris Cité, CNRS, IRIF, F-75013, Paris, France \and \url{} }{ibiapina@irif.fr}{https://orcid.org/0000-0002-6584-7718}{}

\author{Minh Hang {Nguyen}}{Université Paris Cité, CNRS, IRIF, F-75013, Paris, France \and \url{https://www.irif.fr/~mhnguyen/} }{mhnguyen@irif.fr}{https://orcid.org/0009-0008-2391-029X}{}

\author{Mikaël {Rabie}}{Université Paris Cité, CNRS, IRIF, F-75013, Paris, France \and \url{https://www.irif.fr/~rabie/} }{mikael.rabie@irif.fr}{https://orcid.org/0000-0001-6782-7625}{}

\author{Cléophée {Robin}}{Université Paris Cité, CNRS, IRIF, F-75013, Paris, France \and \url{} }{}{https://orcid.org/0000-0002-1825-0097}{}

\authorrunning{A. Ibiapina, M. H. Nguyen, M. Rabie and C. Robin} %TODO mandatory. First: Use abbreviated first/middle names. Second (only in severe cases): Use first author plus 'et al.'

\Copyright{Allen Ibiapina, Minh Hang Nguyen, Mikael Rabie, Cleophee Robin} %TODO mandatory, please use full first names. LIPIcs license is "CC-BY";  http://creativecommons.org/licenses/by/3.0/

\ccsdesc[100]{Mathematics of Computing -> Discrete Mathematics -> Graph Theory-> Graph coloring.} %TODO mandatory: Please choose ACM 2012 classifications from https://dl.acm.org/ccs/ccs_flat.cfm 

\keywords{Temporal Graphs, Coloring.} %TODO mandatoryz; please add comma-separated list of keywords

\category{} %optional, e.g. invited paper

\relatedversion{} %optional, e.g. full version hosted on arXiv, HAL, or other respository/website
\acknowledgements{The authors thank the members of ANR Tempogral and GrR for their insightful inputs, in particular Marthe Bonamy for her contribution with Theorem~\ref{th:static}. This work was supported by the French ANR project TEMPOGRAL (ANR-22-CE48-0001). The second author has received funding from the European Union's Horizon 2020 research and innovation program under the Marie Skłodowska-Curie grant agreement No 945332.}

\nolinenumbers %uncomment to disable line numbering

\begin{document}

\SetVertexStyle[MinSize=0.01, FillColor=black, TextFont=\normalsize, InnerSep=1pt]
\SetEdgeStyle[LineWidth=0.25mm, TextFont= \normalsize]
%This code is for my drawings to work[Allen].

\maketitle

\begin{abstract}
Graph Coloring  consists in assigning colors to vertices ensuring that two adjacent vertices do not have the same color. In dynamic graphs, this notion is not well defined, as we need to decide if different colors for adjacent vertices must happen all the time or not, and how to go from a coloring in one time to the next one.

In this paper, we define a coloring notion for Temporal Graphs where at each step, the coloring must be proper. It uses a notion of compatibility between two consecutive snapshots that implies that the coloring stays proper while the transition happens. Given a graph, the minimum number of colors needed to ensure that such coloring exists is the  \emph{Temporal Chromatic Number} of this graph.

With those notions, we provide some lower and upper bounds for the temporal chromatic number in the general case. We then dive into some specific classes of graphs such as trees, graphs with bounded degree or bounded degeneracy. Finally, we consider temporal graphs where grow pace is one, that is, a single edge can be added and a single other one can be removed between two time steps. In that case, we consider bipartite and bounded degree graphs.

Even though the problem is defined with full knowledge of the temporal graph, our results also work in the case where future snapshots are given online: we need to choose the coloring of the next snapshot after having computed the current one, not knowing what will come later.

\end{abstract}

%Cléophée =>Problemes  : Subsection \ref{trees} 

%%%%%%%%%%%%%%%%%%%%%%%%%%%%%%%%%%%%%%%%%%%%%%%%%%%%%%%%%%%%%%%
\section{Introduction}
%%%%%%%%%%%%%%%%%%%%%%%%%%%%%%%%%%%%%%%%%%%%%%%%%%%%%%%%%%%%%%%

%%%%%%%%%%%%%%%%%%%%%%%%%%%%%%%%%%%%%%%%%%%%%%%%%%%%%%%%%%%%%%%
\subsection{Motivation}
%%%%%%%%%%%%%%%%%%%%%%%%%%%%%%%%%%%%%%%%%%%%%%%%%%%%%%%%%%%%%%%

The Vertex Coloring Problem in (static) graphs is a well-studied problem. The goal is to assign colors to the vertices such that two adjacent vertices have different colors. The usual questions concern the complexity of finding a coloring with few colors, and  the \emph{Chromatic Number} of a graph that is the minimum number of colors needed to ensure the existence of a coloring. The coloring can represent a radio frequency allocation: if two transmitter are too close together (i.e. neighbors in the graph), they cannot use the same frequency, because their messages would collide.

In \emph{Temporal Graphs}, or dynamic networks, changes occur over time. More precisely, at each time step, the set of edges between the vertices evolves. In this work, the time will be discretized. This can be done by considering step times when changes occur. The graph with its set of edges at a fixed time is called a \emph{Snapshot}. To illustrate, temporal graphs can model sensors on a system. At each time step, if two sensors are close enough, an edge between the corresponding vertices appears in the graph. A frequency mapping must ensure that sensors too close to each other must use different wave lengths. This can be modeled by assigning a color to every vertex for every time step such that two vertices adjacent in some snapshot must have different colors in this snapshot. We want to ensure some robustness on the coloring such that, through each transition between two snapshots, no collision can happen.

\begin{comment}
An immediate solution can be to use a proper coloring of the union of all the edges through time. Indeed, with this assignment we are sure that no neighbors have the same color at any time. However, this solution can be expensive considering the number of colors (up to the size $n$ of the network, ). It is possible to add some constraints one the graphs considered to ensure a low chromatic number at each snapshot. For example, a temporal graph such that all snapshot is a tree ensure that all snapshot is 2-colorable (ie. we should be able to use much fewer colors. In this case, it implies that nodes must change their color over time to adapt to their future neighbors, leaving constraints from previous neighbors.
\end{comment}

An immediate solution can be to use a proper coloring of the union of all the edges through time. With such assignment, every pair of vertices that are adjacent in at least one snapshot have  different colors at any time. However, this solution can be expensive considering the number of colors used. Indeed, if every possible edge appears at least once, one different color for every vertex will be used. It is possible to add some constraints on the considered graphs to ensure a low chromatic number for each snapshot. For example, a temporal graph such that all snapshot is a tree ensure that every snapshot is 2-colorable. Such constraint directly implies that it is possible to use less colors. In this case, it implies that nodes must change color over time in order to take into account their future neighborhood while relax constraints from previous neighbors.

\begin{comment}
\stc{Through a transition between two snapshots, some edges could appear or disappear, and some nodes could change  color.
Through a transition, \Cleo{if ?} if we have no control over the order in which those events happen, we need a notion of \emph{Compatibility} between two snapshots. More precisely, for a given node and one of its neighbors, we do not know whether its color will change before the (possible) color change of its neighbor, or whether the edge will disappear first or last in the process. Formally, for two vertices $u$ and $v$, if an edge $uv$ is in the graph at time $i$ or $i+1$, the color of $u$ at time $i$ must be different from the colors of $v$ at times $i$ AND $i+1$. Indeed, during the transition between times $i$ and $i+1$, the edge will be there at some point, and if $v$ changes its color before $u$, the color of $u$ at time $i$ must be different from the color of $v$ at time $i+1$. A direct consequence of this notion is that a coloring of a snapshot must also be a proper coloring of the previous and the next snapshots.}
\end{comment}

Through a transition between two snapshots, three operations can occur: the removal of some edges, the addition of some edges and the color change for some vertices. We consider that those transformations are not done simultaneously but one after the other without any information on the order. More precisely, given two vertices that are adjacent in some snapshot and become non-adjacent in the next snapshot, it is not possible for them to have the same color in the next snapshot. This is because their colors might change before the edge is removed. Another consequence is that two vertices that are adjacent cannot exchange their color as one will have to do the change before the other and between the two, both vertices are still adjacent but have the same colors. Dealing with those transformation motivates the notion of \emph{Compatibility} between two snapshots: for any possible sequence of transformations between two snapshots, the coloring have to stay proper at any time.

A \emph{Temporal Coloring} of a temporal graph is a proper coloring of every snapshot that respects the compatibility constraint. The \emph{Temporal Chromatic Number} of a temporal graph $\cal G$, denoted $\chi^t({\cal G})$, is the minimum number of colors needed for a temporal coloring of the graph. We show how to relate this number to the chromatic number of each snapshot and the chromatic number of the union (or \emph{Smashing}) of three consecutive snapshots.

As we have no control over the evolution of the graph between successive snapshots, we may have massive changes. However, in systems of sensors, not so many changes can occur between two snapshots. We introduce a notion to limit the number of edges that can be added or removed between time steps. The \emph{Grow Pace} of a temporal graph is the maximum number of edges added and removed between two snapshots. For example, if a temporal graph has grow pace 1, that means that, between every consecutive snapshot, at most one edge is added and at most one edge is deleted. In general the grow pace is unbounded. We study and provide temporal chromatic number bounds for temporal graphs with grow pace~1. In particular, we get different temporal chromatic numbers for graphs with grow pace 1 compared to temporal graphs without limitations. A question arising is, given a restriction on the snapshots and a number of colors $k$, what is the maximum grow pace such that a temporal graph admits a temporal $k$-coloring. 

Sometimes, we do not know in advance the full temporal graph. In the case of sensors on a system, we might just be able to anticipate the upcoming moves to know how the connections will be two steps ahead, but not further. An \emph{Online} approach can then be considered. When we compute the coloring of a snapshot, we just know the snapshots up to some time, with no knowledge on the edge set of the next snapshots. More precisely, as we need to follow the compatibility restrictions, to compute the coloring of the snapshot at time $t$, in addition to the previous snapshots, we are given the snapshot at time $t+1$. After that, we will be given the snapshot at time $t+2$ to compute the coloring at time $t+1$, and so on.

\subsection{Related Works}
\noindent \textbf{Temporal Graphs.} Dynamic Networks have been given different names and definitions through time. The important element is to encompass the evolution of the nodes and connections through time. In temporal graphs (see for example~\cite{michail2016intro}), time is discrete and the set of nodes does not change.
It can for example simulate transport systems, and one can wonder about the temporal paths and connectivity~\cite{DBLP:conf/sand/BrunelliV23,DBLP:journals/mst/AkridaGMS17}. Other problems from static graphs have been translated to temporal graphs, such as flows~\cite{DBLP:journals/jcss/AkridaCGKS19} or positional games~\cite{DBLP:conf/sirocco/BalevLLPS20}. 

In our work, we study temporal graphs with limited grow pace. Even with this restriction, temporal graphs can still model concrete situations. 
In the model of Wireless Body Area Networks~\cite{movassaghi2014wireless}, the resulting temporal graphs can have limited grow pace.
%In the Wireless Body Area Networks~\cite{movassaghi2014wireless}, vertices are seen as sensors on a body that move over time, the edges being between two sensors close enough at each snapshot.
%
Usually, the snapshots (i.e. the set of edges at each time) are provided in advance, however there has been some work on online algorithms~\cite{DBLP:journals/corr/abs-2401-16800}, where part of the computation needs to be made while the edges are being given. 

Recently, there are lines of works on designing fully dynamic algorithms for classic graph problems that admit greedy solutions, such as maximal matching \cite{baswana2015fully, solomon2016fully}, maximal independent set \cite{behnezhad2019fully, chechik2019fully}, and $(\Delta+1)$ vertex coloring \cite{bhattacharya2018dynamic, bhattacharya2022fully, henzinger2022constant}. In these works, the main purpose is to maintain a graph property in the presence of edge insertion or deletion such that the update time is minimized. When the sequence of edge insertions and deletions is controlled by an oblivious adversary, there are randomized algorithms for $(\Delta+1)$ vertex coloring problem that spend $O(\log \Delta)$ \cite{bhattacharya2018dynamic}, and $O(1)$ \cite{bhattacharya2022fully, henzinger2022constant} amortized update-time. While in the case of adaptive adversary - meaning the sequence of edge insertions and deletions may depend on the output of algorithms, the update-time is $\tilde{O}(n^{8/9})$ for randomized $(\Delta+1)$-vertex-coloring algorithms~\cite{behnezhad2025fully}.  There are other works on fully dynamic algorithms for $O(\Delta)$ edge coloring \cite{barenboim2017fully} that spends $\tilde{O}(\sqrt{n})$ update-time. Other related work~\cite{barba2017dynamic} is  studying the trade-off between the number of colors used and the number of vertex recolorings.

\vspace{1.5ex}

\noindent \textbf{Temporal Coloring.} Several different definitions for vertex coloring in temporal graphs have been studied.  In~\cite{mertzios2021sliding}, a temporal coloring is defined as a sequence of colorings (not necessarily proper) at different time slots of the given temporal graph such that each edge is properly colored in at least one time slot. A proper sliding $\Delta$-window coloring~\cite{mertzios2021sliding} is a temporal coloring such that every edge is properly colored in at least one time slot that is not too far from the time slot when the edge appear. Later, this temporal coloring problem is studied in $t$-resilient and $t$-occurrent temporal graphs \cite{marino2022coloring}. Another definition was introduced in~\cite{yu2013algorithms}, in which, given a temporal graph, the task is to find a sequence of proper colorings at different time slots and minimize the number of colors used and the total number of times a color is changed.

\vspace{1.5ex}

\noindent \textbf{Distributed Recoloring.} Our definition of temporal coloring is inspired by the notion of recoloring~\cite{DBLP:journals/tcs/BonsmaC09, DBLP:conf/iwpec/0002KKPP14}: given a static graph and two colorings, is there a path of reconfiguration where at each step, one node changes its color, and the coloring remains proper. In our definition, the coloring remains proper at each step, but several nodes change their colors between each snapshot, following a notion of compatibility. The notion of compatibility is directly derived from the notion of parallel recoloring in distributed recoloring~\cite{2018distributed}, although it generalises it. In~\cite{2018distributed}, only an independent set of nodes can change their color at each step, to avoid conflicts during transitions as we cannot control in which order each change occurs. Our compatibility definition generalizes this by allowing more than an independent set of nodes to update their color at the same time, as long as there is no conflict, regardless of the order of the updates.

\subsection{Our Contribution}

Given our definition of \textbf{temporal coloring} in temporal graphs, we  provide results for the general setting, as well as in specific situations, depending on properties of each snapshots, and how snapshots evolve.

Given a temporal graph $\cal G$,  $\chi_s$ denotes the maximum chromatic number over all snapshots of the graph.  We prove that  the temporal chromatic number of $\cal G$ is at most  $ \chi_s^3$ ($\chi^t({\cal G})\leq \chi_s^3$). Furthermore, if all unions of 3 consecutive snapshots can be coloured with at most $\chi_{3s}$, we prove that the temporal chromatic number of $\cal G$ is at least $\chi_{3s}$ and at most $2\chi_{3s}$ colors, i.e. $\chi_{3s}\leq \chi^t({\cal G})\leq 2\chi_{3s} $. In the case where each snapshot is duplicated (i.e. the set of edges change every two snapshots), the upper bound becomes $\chi_s^2$.

We provide a transformation to see our problem as a classical problem of coloring in static graphs, with in particular a simple construction to decide if a temporal graph is temporally 2-colorable.

Given a class of temporal graph we investigate the \emph{maximum chromatic number} of the class, that is the maximum temporal chromatic number among all temporal graphs in that class. We prove the following:
\begin{itemize}
    \item Temporal graphs with each snapshot being a tree have maximal chromatic number between 6 and 8.
    \item Temporal graphs with each snapshot being $d$-degenerate have maximal chromatic number between $5d$ and $12d$.
    \item Temporal graphs with each snapshot being $\Delta$-bounded have maximal chromatic number between $3\Delta +1$ and $5\Delta +1$.
\end{itemize}

We also consider temporal graphs with grow pace 1, that is when at most one edge can be added while at most one edge can be removed between every two consecutive snapshots.  We prove the following:
\begin{itemize}
    \item Temporal graphs with grow pace 1 such that each snapshot being $\Delta$-bounded have maximal chromatic number at most $\Delta +2$. This bound is tight for $\Delta\leq 3$.
\end{itemize}

All the algorithms we provide work online: to compute $c_{i}$, we only use the snapshots up to $G_{i+1}$, without any need to anticipate what will come later.

\subsection{Outline}
We first introduce formally our definition of coloring in temporal graphs, with some other notions we will use in the paper, in Section~\ref{sec:def}. In Sections \ref{sec:general}, we give some upper bounds depending on the chromatic number of each snapshot. In Section~\ref{sec:static}, we show how to solve temporal coloring problems using static graphs. Section~\ref{sec:arbitrary} contains the results on Trees, graphs of bounded degree and graphs of bounded degeneracy. Section~\ref{sec:grow1} is devoted to temporal graphs with grow pace one. %, and Section~\ref{sec:online} deals with temporal graphs where snapshots are given online. 
Finally, Section~\ref{sec:conc} provides some open questions rising from this work.

%%%%%%%%%%%%%%%%%%%%%%%%%%%%%%%%%%%%%%%%%%%%%%%%%%%%%%%%%%%%%%%
\section{Problem Definitions}\label{sec:def}
%%%%%%%%%%%%%%%%%%%%%%%%%%%%%%%%%%%%%%%%%%%%%%%%%%%%%%%%%%%%%%%

%\emph{Static graphs} and their \emph{chromatic number.} 
A static undirected graph $H = (V,E)$ is defined by a vertex set $V$ and an edge set $E$. In this paper we only consider simple and undirected graphs. Trees are connected acyclic graphs, in which two vertices are connected by exactly one path. Paths are trees whose vertex maximum degree is 2. A \emph{$d$-degenerate} graph is a graph where the vertex set can be ordered, say $v_1,\ldots,v_n$, such that each vertex $v_j$ has at most $d$ neighbors in $\{v_{j+1},\ldots v_n\}$. 
A graph is \emph{$\Delta$-bounded} if all of its vertices have degree at most $\Delta$. 

The \emph{chromatic number} of a static graph $G$, denoted by $\chi(G)$, is the minimum number of colors such that one can assign to each vertex of $G$ a color in $\{1,\ldots, \chi(G)\}$ satisfying no two adjacent vertices have the same color. Note that a $d$-degenerate graph can always be colored with $d+1$ colors. Indeed, by coloring vertices $v_i$ from $n$ to $1$, there are at most $d$ forbidden colors when we color $v_i$. It is also always possible to color $\Delta$-bounded graphs with $\Delta+1$ colors.

In this paper we will assume that the time is discrete. Given an integer $T$, the set of integer $\{1,\dots, T\}$ will be denoted by $[T]$.

A \emph{temporal graph} $\mathcal{G}$ is an ordered pair of disjoint sets $(V,E)$ such that $E \subseteq \binom{V}{2} \times \mathbb{N}$. The set $E$ is called the set of time-edges. At a particular time $t \in \mathbb{N}$, $E_t = \{ e:(e,t) \in E\}$ is the (possibly empty) set of all edges that appear in the temporal graph at time $t$. The set $E_t$ can be used to define a \emph{snapshot} of the temporal graph $\mathcal{G}$ at time $t$, which is the static graph $G_t = (V,E_t)$. The natural number $T = \max\{t \mid E_t \neq \emptyset\}$ is said the \emph{lifetime} of $\mathcal{G}$. A temporal graph with lifetime $T$ is an ordered sequence of static graphs $G_1,G_2,\ldots, G_T$. In this paper, we will usually directly use $G_1,G_2,\ldots, G_T$ to describe $\mathcal{G}$.

Let $G' = (V,E')$ and $G'' =(V,E'')$ be two static graphs with the same vertex set $V$. The graph $G' \cup G''$ is the graph whose vertex set is $V $ and edge set is $E' \cup E''$. Let $\mathcal{G}$ be a temporal graph whose snapshots are $G_1,\ldots,G_T$. The $k$-smashed graph of $G_i, \ldots, G_{i+k-1}$ is defined by $G_i \cup \ldots \cup G_{i+k-1}$. In the rest of this paper, $S_i({\cal G})$ denote the 3-smashed graph of $G_{i-1}, G_i, G_{i+1}$.  Hence $S_i({\cal G}) = G_{i-1} \cup G_i \cup G_{i+1} = (V, E_{i-1} \cup E_i \cup E_{i+1})$.

%\Cleo{Used only once in the figure}
%For $i=1,\ldots,T$, let $N_i(v) = \{u|uv \in E_i \}$ be the set of neighbors of vertex $v$ in the snapshot $G_i$. Let $N_{i:j}(v) = \{u| uv \in E_i \cup E_{i+1} \cup \ldots E_j \}$ be a set of neighbors of $v$ in the smashed graph $G_i \cup \ldots \cup G_j$.  In particular, $N_{i:i+1}(v)$ is the set of neighbors of $v$ in the 2-smashed graph $G_i \cup G_{i+1}$.

Let $G_1=(V,E_1)$ and $G_2=(V,E_2)$ be two graphs and let $c_1$ (resp. $c_2$) be a coloring of $G_1$ (resp. $G_2$). We say that $c_1$ and $c_2$ are \emph{compatible} for $G_1$ and $G_2$ if and only if  $\forall uv\in E_1\cup E_2, c_1(u) \neq c_2(v)$. Intuitively, if $v$ changes its color before $u$ while the edge $uv$ is still present, the new color of $v$ must differ from the old color of $u$.

Given a temporal graph $\mathcal{G}=(G_1,G_2,\ldots,G_T)$, a sequence $(c_1,\ldots,c_n)$ is a \emph{$k$-coloring of the temporal graph} $\mathcal{G}$ if,
\begin{enumerate}
    \item For $i \in \{2,\ldots,T-1\}$, $c_i$ is a proper $k$-coloring for $G_{i-1} \cup G_i \cup G_{i+1}$. Moreover, $c_1$ and $c_T$ are proper $k$-colorings of $G_1 \cup G_2$ and $G_{T-1}\cup G_T$ respectively.
    
    \item For $i \in \{1,\ldots,T-1\}$, $c_i$ and $c_{i+1}$ are compatible for $G_i$ and $G_{i+1}$.
\end{enumerate}

The \emph{temporal chromatic number} of a temporal graph $\mathcal{G}$ is the minimum integer $k$ for which there is a $k$-coloring of $\mathcal{G}$, we denote such number by $\chi^t(\mathcal{G})$. In particular, this number is at least as large as the chromatic number of the smashing of any 3 consecutive snapshots, thanks to item 1 of the definition.

We refer to Figure~\Ref{fig:temp_coloring} for an example of a temporal coloring with four consecutive snapshots $G_{1},G_2,G_{3},G_{4}$ using four colors $1,2,3,4$.
The figure illustrates that coloring $c_{3}$ is proper for $3$-smashed graph $G_2 \cup G_{3} \cup G_{4}$. Coloring $c_{3}$ and $c_2$ are compatible in the $2$-smashed graph $G_{2} \cup G_{3}$. 

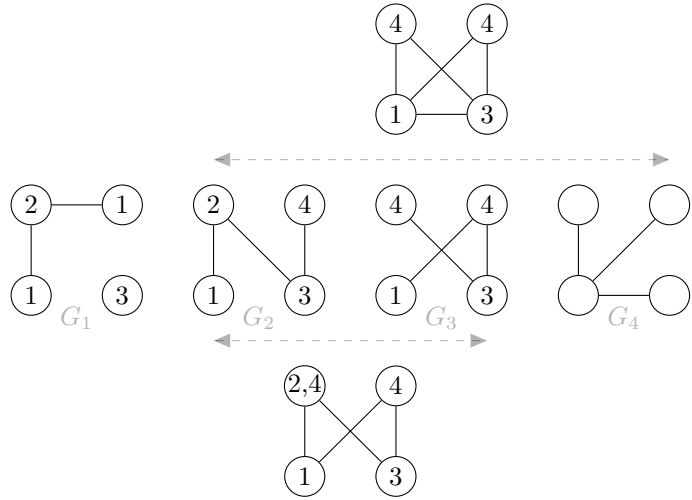
\begin{figure}
    \centering
    \centering
\begin{tikzpicture}[scale=0.6]
\tikzstyle{whitenode}=[draw,circle,minimum size=15pt,inner sep=0pt]
 
\draw (0,0) node[whitenode] (a1)   [] {1};
\draw (2,0) node[whitenode] (a2)   [] {3};
\draw (0,2) node[whitenode] (a3)   [] {2};
\draw (2,2) node[whitenode] (a4)   [] {1};
\draw[opacity=0.3] (1,-0.5) node (a5)   [] {$G_{1}$};

\draw (4,0) node[whitenode] (b1)   [] {1};
\draw (6,0) node[whitenode] (b2)   [] {3};
\draw (4,2) node[whitenode] (b3)   [] {2};
\draw (6,2) node[whitenode] (b4)   [] {4};
\draw[opacity=0.3] (5,-0.5) node (a5)   [] {$G_{2}$};

\draw (8,0) node[whitenode] (c1)   [] {1};
\draw (10,0) node[whitenode] (c2)   [] {3};
\draw (8,2) node[whitenode] (c3)   [] {4};
\draw (10,2) node[whitenode] (c4)   [] {4};
\draw[opacity=0.3] (9,-0.5) node (a5)   [] {$G_{3}$};

\draw (12,0) node[whitenode] (d1)   [] {};
\draw (14,0) node[whitenode] (d2)   [] {};
\draw (12,2) node[whitenode] (d3)   [] {};
\draw (14,2) node[whitenode] (d4)   [] {};
\draw[opacity=0.3] (13,-0.5) node (a5)   [] {$G_{4}$};

\draw (8,4) node[whitenode] (e1)   [] {1};
\draw (10,4) node[whitenode] (e2)   [] {3};
\draw (8,6) node[whitenode] (e3)   [] {4};
\draw (10,6) node[whitenode] (e4)   [] {4};

\draw (6,-4) node[whitenode] (f1)   [] {1};
\draw (8,-4) node[whitenode] (f2)   [] {3};
\draw (6,-2) node[whitenode] (f3)   [] {2,4};
\draw (8,-2) node[whitenode] (f4)   [] {4};

\draw (a1)--(a3)--(a4);
\draw (b1)--(b3)--(b2)--(b4);
\draw (c1)--(c4)--(c2)--(c3);
\draw (d1)--(d2);
\draw (d1)--(d3);
\draw (d1)--(d4);

\draw (e1)--(e3)--(e2)--(e4);
\draw (e2)--(e1)--(e4);

\draw (f1)--(f3)--(f2)--(f4)--(f1);

\draw[>=triangle 45, <->, dashed, opacity=0.3] (4,3) -- (14,3);
\draw[>=triangle 45, <->, dashed, opacity=0.3] (4,-1) -- (10,-1);
    \end{tikzpicture}
    \caption{Temporal coloring $4$ consecutive snapshots $G_{1},G_2,G_{3},G_{4}$. The number $1,2,3,4$ inside every vertex is the color of that vertex in the corresponding snapshots.}
    \label{fig:temp_coloring}
\end{figure}

%An \emph{adaptive adversary} can generate an instance of temporal graph $(G_1,\ldots,G_T)$ sequentially. Moreover, before generating a snapshot $G_i$ at time $i$, the adversary can observe the colorings that an online algorithm give to previous snapshots $G_1,\ldots,G_{i-c}$ for a constant $c \in \mathbb{N}$. \Cleo{Do we use that ? Otherwise, if we want to keep if we can put that in the conclusion}\Mikael{I would keep the Online version, that encapsulate the same idea I believe}

An \emph{online temporal coloring} algorithm $A$ solves Coloring-Temporal-graphs problem with an instance $\mathcal{G}$ generated by an adaptive adversary if there is a constant $c \in \mathbb{N}$ such that for every $i\in [T]$, $A$ can give a coloring $c_i$ for $G_i$ before time $i+c$, depending on partial knowledge about snapshots $G_1,\ldots,G_{i+c}$ of $\mathcal{G}$. All the proofs we provide in this paper do also hold online, that is when the coloring $c_{i}$ is computed without knowledge of snapshots at times $\ge i+2$.

Given a temporal graph $\mathcal{G} = (G_1, \ldots, G_T)$, the \emph{\rate} of $\mathcal{G}$ is defined as the maximum number of edges added or removed between two consecutive snapshots $G_i$ and $G_{i+1}$, for every $i \in [T-1]$. In other words, the \rate~of a graph $\cal G$ is $\max_{i \in [T-1]} \{|E_{i+1} \setminus E_i|, |E_{i} \setminus E_{i+1}| \}$.

%%%%%%%%%%%%%%%%%%%%%%%%%%%%%%%%%%%%%%%%%%%%%%%%%%%%%%%%%%%%%%%
\section{Chromatic number upper bound}\label{sec:general}
%%%%%%%%%%%%%%%%%%%%%%%%%%%%%%%%%%%%%%%%%%%%%%%%%%%%%%%%%%%%%%%
Given a temporal graph, we remind that its temporal chromatic number is lower bounded by the maximum chromatic number of each union of three consecutive snapshots. The two following results define upper bounds of temporal chromatic number, based on the chromatic number of each snapshot, and on the chromatic number of unions of three consecutive snapshots.

\begin{theorem} \label{th:cubeub}
Let $\mathcal{G}=(G_1,\ldots,G_T)$ be a temporal graph. %with vertex set $V$. 
If for all $i\in[T]$, $\chi(G_i) \leq k$, then $\chi^t({\cal G})\leq k^3$.
\end{theorem}
%}
\begin{proof}
Let $K = \{ (i,j,t) \mid i,j,t \in [k] \}$ be a set of $k^3$ colors.

For $i\in\{1,\ldots,T\}$, let $x_i : V\rightarrow [k]$ be a proper $k$-coloring of $G_i$. We define a coloring $c_i : V \rightarrow K$ for $i\in [T]$ as follows:
\begin{itemize}
    \item If $i = 0$ mod 3, let $c_i(v) = (x_{i}(v),x_{i+1}(v),x_{i-1}(v))$

    \item If $i = 1$ mod 3, let $c_i(v) = (x_{i-1}(v),x_{i}(v),x_{i+1}(v))$

    \item If $i = 2$ mod 3, let $c_i(v) = (x_{i+1}(v),x_{i-1}(v),x_{i}(v))$

\end{itemize}

\vspace{0.75ex}

\noindent The sequence $(c_1,\ldots,c_T)$ is a temporal $k^3$-coloring of $\mathcal{G}$ since, for every $i \in [T]$:
\begin{enumerate}
    \item Coloring $c_i$ is a proper coloring for $S_i({\cal G})$. Let $uv \in E_{i-1} \cup E_i \cup E_{i+1}$, then either $x_{i-1}(u) \neq x_{i-1}(v)$ or $x_i(u) \neq x_i(v)$ or $x_{i+1}(u) \neq x_{i+1}(v)$. It implies that $c_i(v) \neq c_i(u)$.
    \item If $i\neq T$, $c_i$ is compatible with $c_{i+1}$ in $E_i \cup E_{i+1}$. Let $uv \in E_i \cup E_{i+1}$, then either $x_i(u) \neq x_i(v)$ or $x_{i+1}(u) \neq x_{i+1}(v)$. It implies that $c_i(v) \neq c_{i+1}(u)$.
\end{enumerate}
\end{proof}

%\begin{theorem}\label{thm:ub2}
%Let $\mathcal{G}=(G_1,\ldots,G_T)$ be a temporal graph. Suppose that for $i \in \{2,\ldots,T-1\}$, the 3-smashed graph of $G_{i-1}$, $G_i$, $G_{i+1}$, denoted by $H_i$, is such that $\chi(H_i) \leq k$. Then, there exists a $2k$-coloring for the temporal graph $\mathcal{G}$.
%\end{theorem}
%\Cleo{
\begin{theorem}\label{thm:ub2}
Let $\mathcal{G}=(G_1,\ldots,G_T)$ be a temporal graph. 
If for all $i \in \{2,\ldots,T-1\}$,  $\chi(G_{i-1}\cup G_i\cup G_{i+1})\leq k$ then $\chi^t(\mathcal{G})\leq 2k$.
\end{theorem}
%}
%\begin{proof}
%For each $i \in \{1, \ldots, T\}$, define a $k$-coloring $c_i$ of $H_i$ as follows: if $i$ is odd, the coloring uses the colors $\{1, \ldots, k\}$; if $i$ is even, the colors used are $\{k+1, \ldots, 2k\}$. By assumption, each $c_i$ is a proper coloring. Moreover, $c_i$ and $c_{i+1}$ are compatible in the graph $G_i \cup G_{i+1}$ since for all $u,v \in V, c_i(u) \neq c_{i+1} (v)$.
%\end{proof}

\begin{proof}
For each $i \in \{1, \ldots, T\}$, let $c_i$ be a $k$-coloring of $G_{i-1}\cup G_i\cup G_{i+1}$ such that if $i$ is odd,  $c_i:V\rightarrow \{1, \ldots, k\}$ and if $i$ is even, $c_i:V\rightarrow\{k+1, \ldots, 2k\}$. By construction, each $c_i$ is a proper coloring of $(G_{i-1}\cup G_i\cup G_{i+1})\leq k$. For all $i \in [T]$, for all $u,v \in V, c_i(u) \neq c_{i+1} (v)$. Hence $c_i$ and $c_{i+1}$ are compatible in the graph $G_i \cup G_{i+1}$ for all $i \in [T]$. Therefore the sequence $(c_1,\ldots,c_T)$ is a temporal $2k$-coloring of $\mathcal{G}$.
\end{proof}

The following result gives an upper bound on the number of colors needed when we double each snapshot. More precisely, if each snapshot is $k$-colorable, the temporal graph needs at most $k^2$ colors, compared to the $k^3$ from Theorem~\ref{th:cubeub}. This motivates the idea to duplicate snapshots (i.e. take two steps) in order to save some colors. In particular, it allows to change colors while the set of edges is stable and keep the same coloring when the set of edges evolves. More formally, duplicating snapshots corresponds to the following constrain: for all $i\in [T]$, if $E_i\sm E_{i+1}\neq \emptyset$ then $E_{i-1}\sm E_{i}= \emptyset$. In other words, changes on the edges set cannot occur two consecutive time. 

\begin{theorem}\label{th:copy}
Let $\mathcal{G} = (G_1$, $G_2$, \ldots, $G_T)$ be a temporal graph with $T$ even in which $G_{2i-1}=G_{2i}$ for all $i\in[T/2]$ (i.e. there are two consecutive copies for each snapshot). If for all $i \in [T]$, $\chi(G_i) \leq k$, then $\chi^t(\mathcal{G})\le k^2$.
\end{theorem}

\begin{proof}
For each $i\in[T/2]$, set $H_i=G_{2i-1}=G_{2i}$. In particular, $\mathcal{G} = (H_1$, $H_1$, $H_2$, $H_2$, $\ldots$, $H_{T/2}$, $H_{T/2})$.

This proof is similar as the proof of Theorem~\ref{th:cubeub}. Let  $K = \{ (i,j) \mid i,j \in [k] \}$ be the set of $k^2$ colors.

For $i\in[T/2]$, let $x_i : V\rightarrow \{1,\dots k\}$ be a proper $k$-coloring of $H_i$. Let  $(c_1^1$, $c_1^2$, $c_2^1$, $c_2^2$, $\ldots$, $c_{T/2}^1$, $c_{T/2}^2)$ be sequence of colors such that:
\begin{enumerate}
    \item For $i \in [T/2]$, if $i$ is even, $c_i^2(v) = (x_{i}(v),x_{i+1}(v))$; if $i$ is odd, $c_i^2(v) = (x_{i+1}(v),x_{i}(v))$; 
    \item Let $c_1^1 = c_1^2$, and for $i\ge2$, $c_i^1 = c_{i-1}^2$.
\end{enumerate}

We claim that $(c_1^1$, $c_1^2$, $c_2^1$, $c_2^2$, $\ldots$, $c_{T/2}^1$, $c_{T/2}^2)$ is a $k^2$-coloring of $\mathcal{G}$. In particular, we prove that the sequence is compatible and that for $i\in [T/2]$ $c^1_i$ is a proper coloring of $S_{2i-1}({\cal G})$ and $c^2_i$ is a proper coloring of $S_{2i}({\cal G})$.

Recall that $\mathcal{G}$ can be seen as the sequence $(H_1$, $H_1$, $H_2$, $H_2$, $\ldots$, $H_{T/2}$, $H_{T/2})$.

\begin{enumerate}
    \item  For $i\in\{2,\ldots,T/2-1\}$, given a vertex $u$, we have that $c_i^1(u)=(x_{i-1}(u),x_i(u))$ or $c_i^1(u)=(x_i(u),x_{i-1}(u))$. Therefore $c_i^1$ is a proper coloring for $H_{i-1}\cup H_{i}$. Directly from the definition, we have that $c_i^2$ is a proper coloring for $H_i \cup H_{i+1}$.  Since  $S_{2i-1}(\mathcal{G})=H_{i-1}\cup H_{i}$ and $S_{2i}(\mathcal{G})=H_{i}\cup H_{i+1}$, it follows that $c^1_i$ is a proper coloring of $S_{2i-1}({\cal G})$ and $c^2_i$ is a proper coloring of $S_{2i}({\cal G})$. Using same arguments, we also have the result for $i=1$ and $i=T/2$. 
    \item For $i\in [T/2]$, since $c_{i}^1=c_{i-1}^2$, they are directly compatible in $H_i\cup H_{i+1}=G_{2i}\cup G_{2i+1}$. Recall that $G_{2i-1}\cup G_{2i}=H_i$. Hence $c^1_i$ and $c^2_i$ are compatible as, for any edge $uv\in H_i$, $x_i(u)\neq x_i(v)$ and $x_i$ appears on the same coordinate in both $c^1_i$ and $c^2_i$. 
\end{enumerate}

Hence, $(c_1^1$, $c_1^2$, $c_2^1$, $c_2^2$, $\ldots$, $c_{T/2}^1$, $c_{T/2}^2)$ is a $k^2$-coloring of $\mathcal{G}$. Therefore $\chi^t(\mathcal{G})\le k^2$.
\end{proof}

%%%%%%%%%%%%%%%%%%%%%%%%%%%%%%%%%%%%%%%%%%%%%%%%%%%%%%%%%%%%%%%
\section{A Link to Static Graphs}\label{sec:static}
%%%%%%%%%%%%%%%%%%%%%%%%%%%%%%%%%%%%%%%%%%%%%%%%%%%%%%%%%%%%%%%
In this section, we show a direct link between temporal coloring and classical coloring. We show how to transform a temporal graph into a static one where both graphs have the same (temporal) chromatic number.

\subsection{2-colorable temporal graphs}

We first start with the case of 2-colorability. In that situation, we provide a property that lets us know when a vertex can change its color. This allows us to create a graph such as, it is bipartite if and only if the initial temporal graph can be temporally 2-colored.

\begin{theorem}\label{thm:2_colorable}
Deciding if a temporal graph is temporally 2-colorable can be done in Polynomial time. 
\end{theorem}
\begin{proof}
Let $\mathcal{G} = (G_1, \ldots, G_T)$ be a temporal graph on vertex set $V$. We start this proof by considering the following result.

\begin{claim}\label{claim:isolated} 
If there exists a temporal 2-coloring $c=(c_1,\dots, c_T)$ of $G$ such that, for $j \in \{2,\ldots,T\}$ and $u \in V$,  $c_j(u) \neq c_{j-1}(u)$, then $u$ has no neighbors in both $G_j$ and $G_{j-1}$. 
\end{claim}
\begin{proofclaim} By the definition of compatibility, if $u$ had a neighbor $v$ in either $ G_{j-1}$ or $G_j$, then $c_{j-1}(u) \neq c_{j}(v)$ and $c_j(u)\neq c_j(v)$. However, since only two colors are allowed, this would violate the coloring rules. Therefore, $u$ must be isolated in both graphs.\end{proofclaim}

From Claim~\ref{claim:isolated}, in order, for a vertex to change color, we must consider times where this vertex is isolated. 
For all vertex $v\in V$, let $(t^v_{1}, \dots, t^v_k)$ be the set of times $t = t^v_i\in [T]$ such that $v$ is isolated in both $G_{t}$ and $G_{t-1}$. We order indices such that $t^v_{i}< t^v_j$ for every $i<j$.  From Claim~\ref{claim:isolated}, it follows that for any temporal 2-coloring of $\mathcal{G}$, the vertex $v$ retains the same color throughout the interval $\{t^v_{i-1}, \ldots,t^v_i - 1\}$ for every $i \in \{0,1,\ldots,k\}$, where $t^v_0 = 1$. 
Recall that, given two integer $i$ and $j$, $[i,j)=\{i,i+1,\ldots,j-1\}$. 

We construct a static graph $col(\mathcal{G})$ as follows: 
\begin{itemize}
    \item $V(col(\mathcal{G}))=\{(v,t^v_i) \mid v\in V\mbox{ and }t^v_i \in \{t^v_0, \ldots, t^v_k\}\}$
    \item  $(u,t^u_i)(v,t^v_j)\in E(col(\mathcal{G}))$ if and only if there exists a time $t \in [t^u_{i}, t^u_{i+1}) \cap [t^v_{j}, t^v_{j+1})$ such that $uv \in E(G_t)$.
\end{itemize}

\begin{figure}
    \centering
    \centering
\begin{tikzpicture}[scale=0.6]

\begin{scope}[shift={(0,0)}]
\draw (-1.5,0) node[] ()   [] {$s$};
\draw (-1.5,1) node[] ()   [] {$w$};
\draw (-1.5,2) node[] ()   [] {$v$};
\draw (-1.5,3) node[] ()   [] {$u$};

\draw (0,0) node[circle,fill, color=blue] (a0)   [] {};
\draw (0,1) node[circle,fill, color=blue] (a1)   [] {};
\draw (0,2) node[circle,fill, color=blue] (a2)   [] {};
\draw (0,3) node[circle,fill, color=blue] (a3)   [] {};
\draw (0,-1) node[] ()   [] {$G_1$};

\draw (2,0) node[circle, draw=orange] (b0)   [] {};
\draw (2,1) node[circle, draw=orange] (b1)   [] {};
\draw (2,2) node[circle, draw=orange] (b2)   [] {};
\draw (2,3) node[circle, draw=orange] (b3)   [] {};
\draw (2,-1) node[] ()   [] {$G_2$};

\draw (4,0) node[circle, draw=teal] (c0)   [] {};
\draw (4,1) node[circle, draw=teal] (c1)   [] {};
\draw (4,2) node[circle, draw=teal] (c2)   [] {};
\draw (4,3) node[circle, draw=teal, fill=teal] (c3)   [] {};
\draw (4,-1) node[] ()   [] {$G_3$};

\draw (6,0) node[circle, draw=red] (d0)   [] {};
\draw (6,1) node[circle, draw=red] (d1)   [] {};
\draw (6,2) node[circle, draw=red, fill=red] (d2)   [] {};
\draw (6,3) node[circle, draw=red] (d3)   [] {};
\draw (6,-1) node[] ()   [] {$G_4$};

\draw (8,0) node[circle, draw=ForestGreen] (e0)   [] {};
\draw (8,1) node[circle, draw=ForestGreen, fill=ForestGreen] (e1)   [] {};
\draw (8,2) node[circle, draw=ForestGreen] (e2)   [] {};
\draw (8,3) node[circle, draw=ForestGreen] (e3)   [] {};
\draw (8,-1) node[] ()   [] {$G_5$};

\draw (10,0) node[circle, draw=black] (f0)   [] {};
\draw (10,1) node[circle, draw=black] (f1)   [] {};
\draw (10,2) node[circle, draw=black] (f2)   [] {};
\draw (10,3) node[circle, draw=black, , fill=black] (f3)   [] {};
\draw (10,-1) node[] ()   [] {$G_6$};

\draw[blue] (a3) edge [out=200,in=160] (a1);
\draw[blue] (a3) edge [out=200,in=160] (a2);
\draw[orange] (b1)--(b0);
\draw[orange] (b2) edge [out=200,in=160] (b0);
\draw[teal] (c1) -- (c0);
\draw[red] (d3) edge [out=220,in=140] (d0);
\draw[ForestGreen] (e2) edge [out=220,in=140] (e0);
\draw[black] (f2) -- (f1);

\draw[fill, opacity =0.1] (-0.5,2.7) rectangle (2.5,3.3);
\draw[fill, opacity =0.1] (3.5,2.7) rectangle (8.5,3.3);
\draw[fill, opacity =0.1] (9.5,2.7) rectangle (10.5,3.3);

\draw[fill, opacity =0.1] (-0.5,1.7) rectangle (4.5,2.3);
\draw[fill, opacity =0.1] (5.5,1.7) rectangle (10.5,2.3);

\draw[fill, opacity =0.1] (-0.5,0.7) rectangle (6.5,1.3);
\draw[fill, opacity =0.1] (7.5,0.7) rectangle (10.5,1.3);

\draw[fill, opacity =0.1] (-0.5,-0.3) rectangle (10.5,0.3);

\draw (4.4,-2) node[] ()   [] {$\cal G$};
\end{scope}

\begin{scope}[shift={(15,0)},scale=1.3]
\draw (0,3) node[rectangle, draw,fill,color= blue] (a0)   [] {};
\node at ([shift={(180:0.8)}]a0) {$(u,1)$};
\draw (2,3) node[rectangle, draw,color=teal, fill] (b0)   [] {};
\node at ([shift={(0:0.8)}]b0) {$(u,3)$};
\draw (4,3) node[rectangle, draw, color=black,fill] (c0)   [] {};
\node at ([shift={(0:0.8)}]c0) {$(u,6)$};

\draw (0,2) node[rectangle, fill,draw,color=blue] (a1)   [] {};
\node at ([shift={(180:0.8)}]a1) {$(v,1)$};
\draw (2,2) node[rectangle, draw, fill, color=red] (b1)   [] {};
\node at ([shift={(0:0.8)}]b1) {$(v,4)$};

\draw (0,1) node[rectangle,fill, draw, color=blue] (a2)   [] {};
\node at ([shift={(180:0.8)}]a2) {$(w,1)$};
\draw (2,1) node[rectangle, draw, color=ForestGreen, fill] (b2)   [] {};  
\node at ([shift={(0:0.8)}]b2) {$(w,5)$};

\draw (0,0) node[rectangle, draw,fill, color=blue] (a3)   [] {};
\node at ([shift={(180:0.8)}]a3) {$(s,1)$};

\draw[blue] (a0) -- (a1);
\draw[blue] (a0) edge [out=-20,in=20] (a2);
\draw[orange] (a2) -- (a3);
\draw[orange] (a1) edge [out=-20,in=20] (a3);

\draw[red] (b0) edge [out=180,in=0] (a3);
\draw[ForestGreen] (b1) edge [out=200,in=-10] (a3);
\draw[black] (b1) -- (b2);

\draw (2,-1.5) node[] ()   [] {$col(\mathcal{G})$};

\end{scope}

\end{tikzpicture}
    \caption{Construction for the proof of Theorem~\ref{thm:2_colorable} }
    \label{fig:twoColarable}
\end{figure}
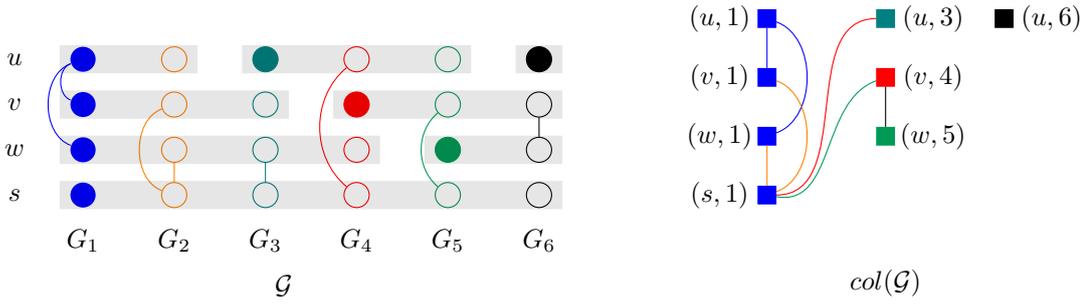

An illustration for the construction of graph $col(\mathcal{G})$ can be found in Figure~\ref{fig:twoColarable}. On the left side of Figure~\ref{fig:twoColarable}, there is a temporal graph $\cal G$ with $6$ snapshots on four vertices $u,v,w,s$. A vertex $v$ at time $t$  is filled if it is isolated at time $t$ and time $(t-1)$ or if $t=1$. The grey bags represent a period of time when a vertex cannot change color in a temporal 2-coloring of $\cal G$. On the right hand side of Figure~\ref{fig:twoColarable}, there is the graph $col(\mathcal{G})$ . The vertex set of $col(\mathcal{G})$ is $\{(u,1),(u,3),(u,6),(v,1),(v,4),(w,1),(w,5),(s,1)\}$ corresponding to the set of grey bags of $\mathcal{G}$. Intuitively, we connect two vertices of $col(\mathcal{G})$ if and only if there is an edge connecting two corresponding bags in $\mathcal{G}$.

\begin{claim}\label{claim:equivalent}
 $\chi^t(\mathcal{G}) \leq 2$ if and only if $\chi(col(\mathcal{G})) \leq 2$.
\end{claim}
\begin{proofclaim}
Suppose that  $\chi^t(\mathcal{G}) \leq 2$ and let $(c_1,\ldots,c_T)$ be a temporal 2-coloring of $\mathcal{G}$. Let $c^*$ be a 2-coloring of $col(\mathcal{G})$ such that $c^*(v, t^v_i) = c_{t^v_i}(v)$ for every vertex $(v,t^v_i)\in V(col(\mathcal{G}))$. We prove that $c^*$ is a proper coloring of $col(\mathcal{G})$. 

Consider an edge $(u,t^u_{i})(v,t^v_{j})$ in $col(\mathcal{G})$. By definition,  there exists a time $t \in [t^u_{i}, t^u_{i+1}) \cap [t^v_{j}, t^v_{j+1})$ such that $uv \in E(G_t)$. Therefore, $c_t(u) \neq c_t(v)$. Furthermore, from Claim~\ref{claim:isolated}, we have that $c_{t^u_i}(u) = c_t(u)$ and $c_{t^v_j}(v) = c_t(v)$. Hence $c^*(u,t^u_{i}) = c_t(v)$, $c^*(v,t^v_{j}) = c_t(v)$, so $c^*(u,t^u_{i}) \neq c^*(v,t^v_{j})$. Therefore $c^*$ is a proper coloring of $col(\mathcal{G})$. This shows that $\chi(col(\mathcal{G})) \leq 2$. 

\vspace{1.5ex}

Conversely, suppose that $\chi(col(\mathcal{G})) \leq 2$ and let $c^*$ be a 2-coloring of $col(\mathcal{G})$. 
For all $t \in [T]$ and all $v \in V$, define $c_t(v) = c^*(v,t^v_i)$, where $t^v_i \leq t$ is the largest integer such that $(v,t^i_v)$ is a vertex in $col(\mathcal{G})$. Observe that, $t=t^v_i$ if and only if $v$ is isolated in $G_{t-1}$ and in $G_t$. We now show that $(c_1,\ldots,c_T)$ is a temporal 2-coloring of $\mathcal{G}$.

Let $uv$ be an edge in $E(G_{t-1} \cup G_t \cup G_{t+1})$ and let $t^u_i,t^v_j \leq t$ be the largest integers such that $(u,t^u_i)$ and $(v,t^v_j)$ are vertices in $col(\mathcal{G})$.
Set ${\cal I}=[t^u_{i}, t^u_{i+1}) \cap [t^v_{j}, t^v_{j+1})$.  We show that there exists $t'\in {\cal I}$ such that $uv\in E(G_{t'})$.

\begin{itemize}
    \item If $uv\in E(G_{t-1})$ then nor $u$ or $v$ are isolated at time $t-1$, hence $t> t^u_i$ and $t>t^v_j$. Therefore $t-1\in {\cal I}$.
    \item If $uv\in E(G_{t})$ then, by the choice of $t^u_i$ and $t^v_j$, $t\in {\cal I}$.
    \item If $uv\in E(G_{t+1})$ then, $t+1\neq t^u_i +1$ and $t+1\neq t^v_j +1$. Therefore $t+1< t^u_{i +1}$ and  $t+1<t^v_{j +1}$. Hence $t+1\in {\cal I}$. 
\end{itemize}

In all cases there exist $t'\in {\cal I}$ such that $uv\in E(G_{t'})$. Therefore $(u,t^u_{i})(v,t^v_{j})\in E(col(\mathcal{G}))$. Hence $c^*(u,t_i^{u}) \neq c^*(v,t_j^{v})$ and so $c_t(u)\neq c_t(v)$.

We still have to consider compatibility. 
Let $u$ be a vertex in $V$ such that, for some $t$, $c_{t-1}(u)\neq c_{t}(u)$. By construction, it exists $t^u_i$ such that $c^*(u,t^u_{i}-1) \neq c^*(u,t^u_{i})$ and $t=t_i^u$. Hence, by definition, $u$ is isolated in $G_{i-1}\cup G_i$. In other words, if $u$ change color between $t-1$ and $t$ then $u$ is isolated in $G_{i-1}\cup G_i$. Therefore, if $uv \in E(G_{t-1} \cup G_t)$, then $c_{t-1}(u)= c_t(v)$ and $c_{t-1}(v) \neq c_t(u)$. 
Suppose $uv \in E(G_{t-1} \cup G_t)$. Recall that, $c_t(v) = c^*(v,t^v_i)$, where $t^v_i \leq t$ is the largest integer such that $(v,t^v_i)$ is a vertex in $col(\mathcal{G})$ and so $v$ is isolated in $G_{t^v_i}$. It implies that $t^v_i\leq t-1$, and $c_{t-1}(v) = c_t(v) = c^*(v,t^v_i)$. Similarly, $c_{t-1}(w) = c_t(w)$. Since either $c_{t-1}(u) \neq c_{t-1}(w)$ or $c_{t}(u) \neq c_{t}(w)$, we have $c_{t-1}(u) \neq c_t(v)$ and $c_{t-1}(v) \neq c_t(u)$. This completes the proof that $(c_1,\ldots,c_T)$ is a proper 2-coloring of $\mathcal{G}$.\end{proofclaim}

To conclude the proof of the theorem, we just have to note that, to decide if $\cal G$ is temporally 2-colorable we can build $col(\mathcal{G})$ that is a graph with at most $|V|\times |T|$ vertices. Then, by claim~\ref{claim:equivalent},  $col(\mathcal{G})$ is bipartite if and only if $\cal G$ is temporally 2-colorable. Since deciding if a static graph is bipartite can be done in Polynomial time, this prove the Theorem.
\end{proof}

\subsection{From Temporal Coloring to Classic Coloring}

The problem of finding a $k$-coloring for a temporal graph $\mathcal{G}$ can be rewritten equivalently to a classical coloring problem for a static graph $static(\mathcal{G})$, which is constructed as follows: 
\begin{itemize}
    \item $V_{static} =  \{ v_i \mid v \in V, i \in [T]\} $;
    
    \item $E_{static} = E' \cup E''$ with:
    $$E' = \bigcup_{i=1}^{T} \{(u_i,v_i) \mid (u_i,v_i) \in E( G_{i-1} \cup G_i \cup G_{i+1}) \}$$
    $$E'' = \bigcup_{i=1}^{T-1} \{(u_i,v_{i+1}) \mid (u_i,v_i) \in E(G_i \cup G_{i+1}) \} $$
    
\end{itemize}

Intuitively, we duplicate each node for each snapshot. The color of $v$ in snapshot $G_i$ will be the same as the color of $v^i$ in $static(\mathcal{G})$. The edges in $E'$ ensure that the color of $v^i$ is proper with its neighbors in $G_{i-1}\cup G_i\cup G_{i+1}$. The edges in $E''$ are equivalent to the compatibility condition.

\begin{theorem}\label{th:static}
    For any temporal graph $\mathcal{G}$, we have $\chi^t(\mathcal{G})=\chi(static(\mathcal{G}))$.
\end{theorem}

\begin{proof}
If there exists a $k$-coloring $(c_1,\ldots,c_T)$ for a temporal graph $\mathcal{G}$, then we can find a proper $k$-coloring, $C$, for the corresponding ${G}_{static(\mathcal{G})}$ by setting $C(v^i) = c_i(v)$. By the compatibility property of $(c_1,\ldots,c_T)$, the configuration $C$ creates no conflict on the edge set $E''$. Moreover an edge $u^iv^i\in E'$ corresponds to some edge $uv$ in $G_{i-1}\cup G_i\cup G_{i+1}$, ensuring that this edge is properly colored.

Conversely, if there exists a $k$-coloring $C$ for $static(\mathcal{G})$, then, we can choose $c_i(v) = C(v^i)$, for all $v \in V$, $i \in [T]$. The coloring
$(c_1,\ldots,c_T)$ is a $k$-coloring for the original temporal graph for similar reasons.
\end{proof}

Unfortunately, analyzing directly those graphs does not, yet, give results for the specific cases we consider in the following sections.

%%%%%%%%%%%%%%%%%%%%%%%%%%%%%%%%%%%%%%%%%%%%%%%%%%%%%%%%%%%%%%%
\section{Arbitrary Grow Pace}\label{sec:arbitrary}
%%%%%%%%%%%%%%%%%%%%%%%%%%%%%%%%%%%%%%%%%%%%%%%%%%%%%%%%%%%%%%%

In this section, we restrict all snapshots to some classical  classes of graphs. We choose classes where the coloring questions are well known. The classes we consider are trees, paths, graphs with bounded degeneracy and graphs with bounded degree. For all those cases, we provide better bounds for the temporal chromatic number, in particular through the analysis of the smashing of 3 snapshots.

We start with a quick analysis for the case where each snapshot is a bipartite graph.

%%%%%%%%%%%%%%%%%%%%%%%%%%%%%%%%%%%%%%%%%%%%%%%%%%%%%%%%%%%%%%%
\subsection{Bipartite Graphs}\label{sec:bipartite}
%%%%%%%%%%%%%%%%%%%%%%%%%%%%%%%%%%%%%%%%%%%%%%%%%%%%%%%%%%%%%%%

Given a temporal graph where each snapshot is bipartite, we know that the maximum chromatic number for each snapshot is 2. Hence, from Theorem~\ref{th:cubeub}, we directly have that the temporal chromatic number is at most 8. Theorem~\ref{t:bip_8} show that this bound is tight. 

\begin{theorem}\label{t:bip_8}
    The maximal chromatic number of temporal graphs where each snapshot is a bipartite graph is 8.
\end{theorem}

\begin{proof}
    Theorem~\ref{th:cubeub} for bipartite graphs implies that any temporal graph where each snapshot is bipartite can be temporally 8-colored. To conclude the proof, we need to provide a temporal graph with temporal chromatic number 8.
    
    In Figure~\ref{f:bip8}, there is a temporal graph on 8 vertices with 3 snapshots such that all snapshots are bipartite graphs (bi-partition given by dark and white vertices).  One can see that $S_2({\cal G})$ is a complete graph. Hence $\chi ( S_2({\cal G}))\ge8$ and so $\chi^t({\cal G})=8$.
\end{proof}

%\Cleo{Observe that, with at least 9 vertices, it is not possible to find 3 snapshots $G_1$, $G_2$ and $G_3$, all bipartite such that the smashing $G_1\cup G_2\cup G_3$ is a complete graphs. Indeed, if there is more than 9 vertices, by the pigeon hole principal, in $G_1$ at least 5 vertices belongs to the same side of the bipartition. Among those 5 vertices, at least 3 of them are in the same  side of the bipartition in $G_2$. Among those 3 vertices, at least 2 of them are in the same  side of the bipartition in $G_3$. This is consistent with the fact that 8 colors are always sufficient. }\Cleo{En vraie, je trouve ça joli même si pas super utile. ON peut le retirer si besoin. }\Mikael{It is a direct corollary to the fact that it is 8-colorable, I prefer to remove it}

\begin{figure}
      \begin{tikzpicture}
  			
    \begin{scope}[scale=0.85]
			\node[vertex,fill] (a) at (0,0) {};
			\node[vertex] (b) at (1.7,0) {};
			\node[vertex,fill] (c) at (2.55,1) {};
			\node[vertex] (d) at (2.55,2.2) {};
			\node[vertex,fill] (e) at (1.7,3.2) {};
			\node[vertex] (f) at (0,3.2) {};
			\node[vertex,fill] (g) at (-0.8,2.2) {};
			\node[vertex] (h) at (-0.8,1) {};
		
		    \node[] (name) at (0.6,-1) {$G_1$};

			\draw[color=black] (a) -- (b);
			\draw[color=black] (a) -- (d);
			\draw[color=black] (a) -- (f);
			\draw[color=black] (a) -- (h);
			\draw[color=black] (b) -- (c);
			\draw[color=black] (b) -- (e);
			\draw[color=black] (b) -- (g);
			\draw[color=black] (c) -- (d);
			\draw[color=black] (c) -- (f); 
			\draw[color=black] (c) -- (h); 
			\draw[color=black] (d) -- (e);
			\draw[color=black] (d) -- (g);
			\draw[color=black] (e) -- (f);
			\draw[color=black] (e) -- (h);
			\draw[color=black] (f) -- (g); 
			\draw[color=black] (h) -- (g); 
	\end{scope}

 \begin{scope}[xshift=5cm,yshift=0cm,scale=0.85]
 
 	\node[vertex,fill] (a) at (0,0) {};
			\node[vertex,fill] (b) at (1.7,0) {};
			\node[vertex] (c) at (2.55,1) {};
			\node[vertex] (d) at (2.55,2.2) {};
			\node[vertex,fill] (e) at (1.7,3.2) {};
			\node[vertex,fill] (f) at (0,3.2) {};
			\node[vertex] (g) at (-0.8,2.2) {};
			\node[vertex] (h) at (-0.8,1) {};
		
              \node[-] (name) at (0.6,-1) {$G_2$};
              
			\draw[color=black] (a) -- (c);
			\draw[color=black] (a) -- (g);
			\draw[color=black] (b) -- (d);
			\draw[color=black] (b) -- (h);
		
			\draw[color=black] (c) -- (e);
 
			\draw[color=black] (d) -- (f);

			\draw[color=black] (e) -- (g);
			\draw[color=black] (f) -- (h); 

	\end{scope}

   \begin{scope}[xshift=10cm,yshift=0cm,scale=0.9]
 
 	\node[vertex,fill] (a) at (0,0) {};
			\node[vertex,fill] (b) at (1.7,0) {};
			\node[vertex,fill] (c) at (2.55,1) {};
			\node[vertex,fill] (d) at (2.55,2.2) {};
			\node[vertex] (e) at (1.7,3.2) {};
			\node[vertex] (f) at (0,3.2) {};
			\node[vertex] (g) at (-0.8,2.2) {};
			\node[vertex] (h) at (-0.8,1) {};
		
              \node[] (name) at (0.6,-1) {$G_3$};
              
			\draw[color=black] (a) -- (e);
			\draw[color=black] (b) -- (f);
			\draw[color=black] (c) -- (g);
			\draw[color=black] (d) -- (h);

	\end{scope}

		\end{tikzpicture}
    \caption{A temporal graph with all snapshots being bipartite and temporal chromatic number 8.}
    \label{f:bip8}
\end{figure}
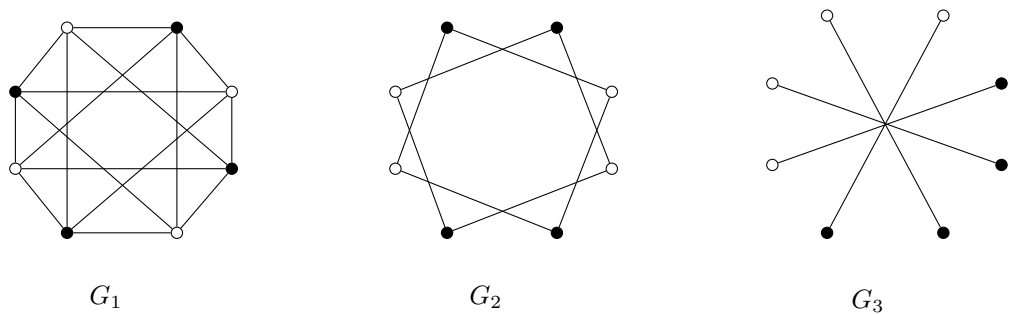

%%%%%%%%%%%%%%%%%%%%%%%%%%%%%%%%%%%%%%%%%%%%%%%%%%%%%%%%%%%%%%%
\subsection{Trees and Paths}\label{trees}
%%%%%%%%%%%%%%%%%%%%%%%%%%%%%%%%%%%%%%%%%%%%%%%%%%%%%%%%%%%%%%%

Let ${\cal G}=(G_1,\ldots,G_T)$ be a temporal graph whose snapshots are trees. We start by the case where we duplicate each snapshot, and get tight bounds. Theorem~\ref{th:copy} shows that there is a temporal 4-coloring for a such temporal graph. Moreover, there is a way to decompose a complete graph on 4 vertices into two paths, see Figure~\ref{fig:tree4}. It implies that there is a temporal graph, where we duplicate each snapshot, and the chromatic number of this temporal graph is $4$.

For the general case, where we do not duplicate each tree, we can just define bound the maximal chromatic number that are possibly not tight.

\begin{lemma}\label{lem:tree_lb}
Let $S_i$ be the three-smashed graph of any three arbitrary trees $G_{i-1}$, $G_i$ and $G_{i+1}$. We have $\chi(S_i)\le6$. Moreover, this bound is tight.
\end{lemma}
\begin{proof}
Let $H = (V_H,E_H)$ be any arbitrary induced subgraph of $S_i$. Since the induced subgraphs on the vertex set $V_H$ of $G_{i-1}$, $G_i$, and $G_{i+1}$ are trees, 
it holds that $|E_H| \leq 3(|V_H|-1)$. The sum of vertex degrees in $H$, being equal to twice the number of edges, is strictly less than $6|V_H|$. Therefore, $H$ must have a vertex of degree at most 5 for every induced subgraph $H$ of $S_i$. It implies that $S_i$ is a 5-degenerate graph. Hence, $G$ is 6-colorable.

Furthermore, the bound is tight. Figure~\ref{smtree} shows an example such that $S_2=G_1\cup G_2\cup G_3$ is a clique of size 6 and $G_1$, $G_2$ and $G_3$ are paths. 
\end{proof}

\begin{figure}[ht]
\centering
\begin{minipage}{0.44\textwidth}
        \centering
    \begin{tikzpicture}[scale=0.7]
\tikzstyle{whitenode}=[draw,circle,minimum size=12pt,inner sep=0pt]
\tikzstyle{nonode}=[draw=white,circle,minimum size=12pt,inner sep=0pt]

    \draw (0,3.2) node[nonode] (nn1)   [] {};
    \draw (0,-1.2) node[nonode] (nn2)   [] {};
    
    \draw (0,0) node[whitenode] (a)   [] {};
    \draw (0,2) node[whitenode] (b)   [] {};
    \draw (2,2) node[whitenode] (c)   [] {};
    \draw (2,0) node[whitenode] (d)   [] {};
    
      \draw (a) edge[red, thick, dashed] (b);
      \draw (b) edge[red, thick, dashed] (c);
      \draw (c) edge[red, thick, dashed] (d);
      \draw (c) edge[blue, thick, dotted] (a);
      \draw (a) edge[blue, thick, dotted] (d);
      \draw (d) edge[blue, thick, dotted] (b);
    \end{tikzpicture}
    \caption{Two paths: $H_1$ (dotted blue line), $H_2$~(dashed red line), such that $H_1\cup H_2$ is a clique.}\label{fig:tree4}
\end{minipage}%
\begin{minipage}{0.05\textwidth}
~
\end{minipage}
\begin{minipage}{0.44\textwidth}
        \centering
    \begin{tikzpicture}[scale=0.4]
      \graph[circular placement, radius=2.0cm,
             empty nodes, nodes={circle,draw}] {
        \foreach \x in {a,...,f} {
            \x
          };
        };
      
      %\foreach \x [count=\idx from 0] in {a,...,f} {
       % \pgfmathparse{90 + \idx * (360 / 6)}
       % \node at (\pgfmathresult:4.4cm) {\x};
      %};
      \draw (a) edge[red, thick, dashed] (f);
      \draw (a) edge[red, thick, dashed] (b);
      \draw (b) edge[red, thick, dashed] (c);
      \draw (c) edge[red, thick, dashed] (d);
      \draw (d) edge[red, thick, dashed] (e);
      \draw (a) edge[blue, thick, dotted] (c);
      \draw (c) edge[blue, thick, dotted] (e);
      \draw (e) edge[blue, thick, dotted] (f);
      \draw (f) edge[blue, thick, dotted] (b);
      \draw (a) edge[blue, thick, dotted] (d);
      \draw (c) edge[black, thick] (f);
      \draw (f) edge[black, thick] (d);
      \draw (d) edge[black, thick] (b);
      \draw (b) edge[black, thick] (e);
      \draw (e) edge[black, thick] (a);
    \end{tikzpicture}
    \caption{Three consecutive snapshots: $G_1$ (dotted blue line), $G_2$ (dashed red line), $G_3$ (solid black line), which are paths. The 3-smashed graph $G_1 \cup G_2 \cup G_3$ is the complete graph on six vertices.}\label{smtree}
\end{minipage}
\end{figure}

The following theorem is a direct consequence from Theorem~\ref{th:cubeub} and Lemma~\ref{lem:tree_lb}.

\begin{theorem}\label{thm:tree}
    The maximal chromatic number of temporal graphs with each snapshot being a tree is between 6 and 8.
\end{theorem}

%%%

%There exists temporal graph whose snapshots are trees such that its temporal chromatic number is 6. On the other hand, since the chromatic number of each snapshot is $2$, Theorem~\ref{th:cubeub} implies that the chromatic number of a temporal tree is at most 8. A direct 
We leave as an open question: Does there exist temporal graphs where each snapshot is a tree with temporal chromatic number 7 or 8?

\vspace{2ex}

The following result is about online temporal coloring, and give some information for the open question whether 7 colors are needed to color a temporal tree. We show that there exists some 6-coloring $c_i$ for some well chosen snapshots up to $G_{i+1}$ and some $G_{i+2}$ such that any coloring $c_{i+1}$ needs at least 7 colors. It implies that, to compute an online temporal 6-coloring (if it is possible), choosing some $c_i$ that is just a coloring compatible with $c_{i-1}$ is not enough.

\begin{figure}[ht]
    \centering
    \begin{tikzpicture}[scale=0.85]
\tikzstyle{whitenode}=[draw,circle,minimum size=12pt,inner sep=0pt]

\begin{scope}[shift={(0,0)}]
    
\draw (0,0) node[whitenode] (u6)   [] {$u_6$};
\draw (0.4,0) node[] ()   [above] {$1$};
\draw (0,1) node[whitenode] (u5)   [] {$u_5$};
\draw (0.4,1) node[] ()   [above] {$1$};
\draw (0,2) node[whitenode] (u4)   [] {$u_4$};
\draw (0.4,2) node[] ()   [above] {$2$};
\draw (0,3) node[whitenode] (u3)   [] {$u_3$};
\draw (0.4,3) node[] ()   [above] {$3$};
\draw (0,4) node[whitenode] (u2)   [] {$u_2$};
\draw (0.4,4) node[] ()   [above] {$4$};
\draw (0,5) node[whitenode] (u1)   [] {$u_1$};
\draw (0.4,5) node[] ()   [above] {$5$};
\draw (-1,6) node[whitenode, fill=red!50] (a)   [] {$A$};
\draw (-0.6,6) node[] ()   [above] {$6$};

\draw (1.8,1) node[whitenode] (h1)   [] {$w_1$};
\draw (2.2,1) node[] ()  [above] {$3$};
\draw (3.7,1) node[whitenode] (h2)   [] {$w_2$};
\draw (4.2,1) node[] ()  [above] {$3$};
\draw (2.7,0.25) node[whitenode] (h3)   [] {$w_3$};
\draw (3.1,0.25) node[] ()   [above] {$4$};
\draw (1.8,5) node[whitenode] (h4)   [] {$w_4$};
\draw (2.2,5.25) node[] ()   [above] {$6$};
\draw (3.7,5) node[whitenode] (h5)   [] {$w_5$};
\draw (4.2,5.25) node[] ()   [above] {$6$};

\draw (5.5,0) node[whitenode] (v6)   [] {$v_6$};
\draw (6,0) node[] ()   [above] {$1 $};
\draw (5.5,1) node[whitenode] (v5)   [] {$v_5$};
\draw (6,1) node[] ()   [above]{$1$};
\draw (5.5,2) node[whitenode] (v4)   [] {$v_4$};
\draw (6,2) node[] ()   [above] {$2$};
\draw (5.5,3) node[whitenode] (v3)   [] {$v_3$};
\draw (6,3) node[] ()   [above] {$3$};
\draw (5.5,4) node[whitenode] (v2)   [] {$v_2$};
\draw (6,4) node[] ()   [above] {$4$};
\draw (5.5,5) node[whitenode] (v1)   [] {$v_1$};
\draw (6,5) node[] ()   [above] {$5$};
\draw (6.5,6) node[whitenode,fill=red!50] (b)   [] {$B$};
\draw (6,6) node[] ()   [above] {$6$};

\draw[dotted,red, thick] (h4) edge (h1);
\draw[dotted,red, thick] (u5) edge [out=30,in=-130] (h4);
\draw[dotted,red, thick] (u1) edge [out=-40,in=50] (u5);
\draw[dotted,red, thick] (u1) edge (u2);
\draw[dotted,red, thick] (a) edge [out=-80,in=125] (u2);
\draw[dotted,red, thick] (a) edge [out=-85,in=130] (u3);
\draw[dotted,red, thick] (u3) edge (u4);
\draw[dotted,red, thick] (u4) edge [out=-130,in=125] (u6);
\draw[dotted,red, thick] (u6) edge (h3);

\draw[dotted,red, thick] (h5) edge (h2);
\draw[dotted,red, thick] (v5) edge [out=150,in=-50] (h5);
\draw[dotted,red, thick] (v1) edge [out=-140,in=130] (v5);
\draw[dotted,red, thick] (v1) edge (v2);
\draw[dotted,red, thick] (b) edge [out=-100,in=55] (v2);
\draw[dotted,red, thick] (b) edge [out=-95,in=50] (v3);
\draw[dotted,red, thick] (v3) edge (v4);
\draw[dotted,red, thick] (v4) edge [out=-50,in=55] (v6);
\draw[dotted,red, thick] (v6) edge  (h3);

\draw[opacity=0.5] (2.7,-0.5) node[] ()  [below] {$G_2\cup G_3$};

\draw[teal] (a) edge [out=-90,in=135] (u4);
\draw[teal] (a) edge [out=-95,in=130] (u6);
\draw[teal] (u2) edge [out=-130,in=130] (u4);
\draw[teal] (u3) -- (u2);
\draw[teal] (u1) edge [out=-40,in=40] (u3);
\draw[teal] (u1) -- (h4);

\draw[teal] (u6) -- (h1) -- (u5) -- (h3) -- (v5) -- (h2) -- (v6);

\draw[teal] (b) edge [out=-90,in=45] (v4);
\draw[teal] (b) edge [out=-85,in=50] (v6);
\draw[teal] (v2) edge [out=-50,in=50] (v4);
\draw[teal] (v3) -- (v2);
\draw[teal] (v1) edge [out=-130,in=130] (v3);
\draw[teal] (v1) -- (h5);

\end{scope}

\begin{scope}[shift = {(9.1,0)}]
     
\draw (0,0) node[whitenode] (u6)   [] {$u_6$};
\draw (0.7,0) node[] ()  [above]{$1 \rightarrow 5$};
\draw (0,1) node[whitenode] (u5)   [] {$u_5$};
\draw (0.7,1) node[] ()   [above] {$1 \rightarrow 2$};
\draw (0,2) node[whitenode] (u4)   [] {$u_4$};
\draw (0.7,2) node[] ()   [above] {$2 \rightarrow 5$};
\draw (0,3) node[whitenode] (u3)   [] {$u_3$};
% \draw (0.3,3) node[] ()   [below] {$3$};
\draw (0,4) node[whitenode] (u2)   [] {$u_2$};
% \draw (0.3,4) node[] ()   [below] {$4$};
\draw (0,5) node[whitenode] (u1)   [] {$u_1$};
% \draw (0.3,5) node[] ()   [below] {$5$};
\draw (-1,6) node[whitenode, fill=red!50] (a)   [] {$A$};
\draw (-0.25,6) node[] ()   [above] {$6 \rightarrow 5$};

\draw (1.8,1) node[whitenode] (h1)   [] {$w_1$};
%\draw (2.2,1) node[] ()  [above] {$3$};
\draw (3.7,1) node[whitenode] (h2)   [] {$w_2$};
%\draw (4.2,1) node[] ()  [above] {$3$};
\draw (2.7,0.25) node[whitenode] (h3)   [] {$w_3$};
%\draw (3.2,0.25) node[] ()   [above] {$4$};
\draw (1.8,5) node[whitenode] (h4)   [] {$w_4$};
%\draw (2.2,5.25) node[] ()   [above] {$6$};
\draw (3.7,5) node[whitenode] (h5)   [] {$w_5$};
%\draw (4.2,5.25) node[] ()   [above] {$6$};

\draw (5,0) node[whitenode] (v6)   [] {$v_6$};
\draw (5.8,0) node[] ()   [above] {$1 \rightarrow 5$};
\draw (5,1) node[whitenode] (v5)   [] {$v_5$};
\draw (5.8,1) node[] ()   [above] {$1 \rightarrow 2$};
\draw (5,2) node[whitenode] (v4)   [] {$v_4$};
\draw (5.8,2) node[] ()   [above] {$2 \rightarrow 5$};
\draw (5,3) node[whitenode] (v3)   [] {$v_3$};
% \draw (6.6,3) node[] ()   [below] {$3$};
\draw (5,4) node[whitenode] (v2)   [] {$v_2$};
% \draw (6.6,4) node[] ()   [below] {$4$};
\draw (5,5) node[whitenode] (v1)   [] {$v_1$};
% \draw (6.6,5) node[] ()   [below] {$5$};
\draw (6,6) node[whitenode,fill=red!50] (b)   [] {$B$};
\draw (5.25,6) node[] ()   [above]{$6 \rightarrow 5$};

\draw[dashed,blue] (a) -- (u1) -- (u2) -- (u3) -- (u4) -- (u5) -- (u6);
\draw[dashed,blue] (b) -- (v1) -- (v2) -- (v3) -- (v4) -- (v5) -- (v6);
\draw[dashed,blue] (a) -- (b);
\draw[dashed,blue] (h1) -- (h4) -- (h3) -- (h5) -- (h2);
\draw[dashed,blue] (u6) edge [out=0,in=-100] (h1);

\draw[dotted,red, thick, opacity=0.3] (a) edge [out=-80,in=125] (u2);
\draw[dotted,red, thick, opacity=0.3] (a) edge [out=-85,in=130] (u3);
\draw[teal, opacity=0.3] (a) edge [out=-90,in=125] (u4);
\draw[teal, opacity=0.3] (a) edge [out=-95,in=130] (u6);

\draw[dotted,red, thick, opacity=0.3] (b) edge [out=-100,in=55] (v2);
\draw[dotted,red, thick, opacity=0.3] (b) edge [out=-95,in=50] (v3);
\draw[teal, opacity=0.3] (b) edge [out=-90,in=45] (v4);
\draw[teal, opacity=0.3] (b) edge [out=-85,in=50] (v6);

\draw[opacity=0.5] (2.8,-0.5) node[] ()   [below] {$G_4$};
\end{scope}

\end{tikzpicture}
    \caption{Three consecutive paths: a part of $G_2$ (dotted lines), a part of $G_3$ (solid lines), a part of $G_4$ (dashed lines). A node $v$ labeled $a$ means $c_2(v)=a$; a node $v$ labeled $a \rightarrow b$ means that $c_2(v) = a$ and $\{u \mid uv \in E_2 \cup E_3\}=[6]\setminus \{a,b\}$.}\label{paths}
\end{figure}
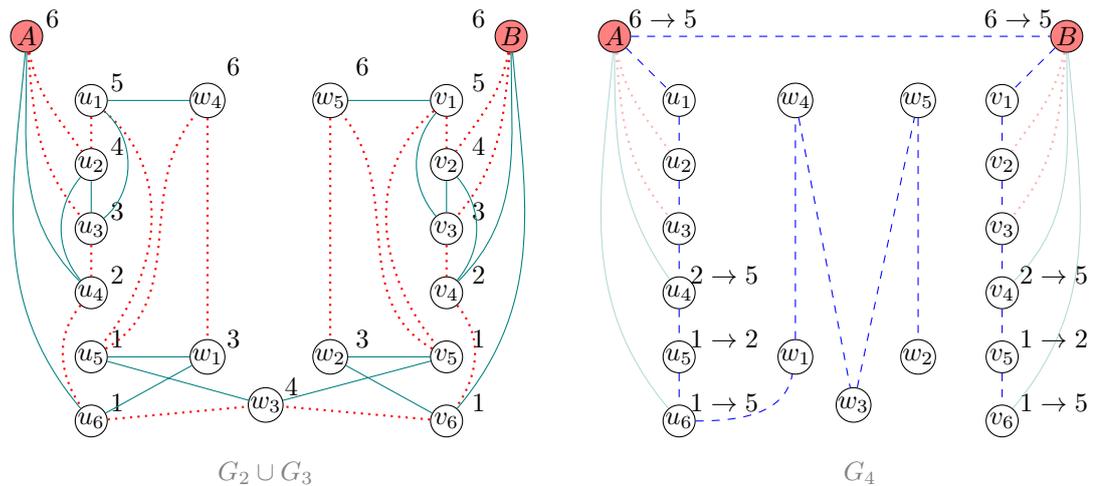

\begin{theorem}\label{th:7colors}
There exists a temporal graph $\mathcal{G}$ whose snapshots are paths, and a coloring $c_i$ for $S_i$ using 6 colors such that any coloring $c_{i+1}$ for $S_{i+1}$ needs at least 7 colors.
\end{theorem}

\begin{proof}

Let $\cal G$ be a temporal graph on lifetime 4 whose snapshots $G_1$, $G_2$, $G_3$ and $G_4$ are paths. Let $\{A,u_1,u_2,u_3,u_4,u_5,u_6,B,v_1,v_2,v_3,v_4,v_5,v_6,w_1,w_2,w_3,w_4,w_5\}$ be 19 vertices of $\cal G$. The description of $\cal G$ can be found in Figure~\ref{paths}. On the left hand side of the figure, path $G_2$ is presented by dotted red line, $G_3$ is solid teal line. On the right hand side, $G_4$ is dashed blue line.

We prove that there exits a coloring $c_2$ on $G_2$ with $6$-colors such that there is no coloring $c_3$ on $G_3$ that is a proper coloring of $G_2\cup G_3\cup G_4$, and $c_3$ is compatible with $c_2$ on  $G_2\cup G_3$.

Let $c_2$ be a coloring of $G_2$ such that : 
\begin{itemize}
    \item $c_2(A)=c_2(B)=c_2(w_4)=c_2(w_5)=6$,
    \item $c_2(u_1)=c_2(v_1)=4$,
    \item $c_2(u_2)=c_2(v_2)=c_2(w_3)=4$,
    \item $c_2(u_3)=c_2(v_3)=c_2(w_1)=c_2(w_2)=3$,
    \item $c_2(u_4)=c_2(v_4)=2$,
    \item $c_2(u_5)=c_2(v_5)=1$,
    \item $c_2(u_6)=c_2(v_6)=1$.
\end{itemize}

Assume by contradiction that there exists a satisfying coloring $c_3$ using $6$ colors. Note that in $G_2 \cup G_3$, vertices $u_4,u_5,u_6,v_4,v_5,v_6,A,B$ have degree $4$.
Since $c_3$ needs to be compatible with $c_2$ in $G_2 \cup G_3$, for every vertex $s$ in $\{u_4,u_5,u_6,v_4,v_5,v_6,A,B\}$, there are two colors that can be used for $c_3(s)$. On the right hand side of Figure~\ref{paths}, a vertex $s$ labeled by $a \rightarrow b$ means that $s$ can be colored by either $a$ or $b$ in $c_3$. As an example, $c_3(u_6)\in \{1;5\}$ because color 2 is forbidden as $u_6u_4\in E_2$,  color 3 is forbidden as $u_6w_1\in E_3$,  color 4 is forbidden as $u_6w_3\in E_2$ and color 6 is forbidden as $u_6a\in E_3$.

Since vertex $u_5$ must have color 1 or 2 in $c_3$, either $u_4$ or $u_6$ have to change its color to $5$ in $c_3$. Since either $u_4$ or $u_6$ has color $5$ in $c_3$, it implies that $c_3(A)\neq 5$ and so $c_3(A)=6$. By the symmetry of the graph, with similar arguments, we have $c_3(B)=6$.

This gives a contradiction since $A$ and $B$ are adjacent in $G_4$. 
\end{proof}

%%%%%%%%%%%%%%%%%%%%%%%%%%%%%%%%%%%%%%%%%%%%%%%%%%%%%%%%%%%%%%%
\subsection{Bounded Degeneracy Graphs} \label{degenerate}
%%%%%%%%%%%%%%%%%%%%%%%%%%%%%%%%%%%%%%%%%%%%%%%%%%%%%%%%%%%%%%%

In this section, let $\mathcal{G}=(G_1,\ldots,G_T)$ be a temporal graph whose snapshots are all $d$-degenerate graphs. 

Recall that a graph is $d$-degenerate if all its subgraphs contain a vertex with degree at most $d$. Observe that $d$-degenerate graphs $G$ are characterized by a degenerate ordering that is a linear ordering on its vertices $v_1 \leq v_2 \leq \dots \leq v_n$ such that for all $i\in [n]$, $v_i$ has at most $d$ neighbors in $\{v_{i+1},\dots, v_n\}$.

For $i\in \{2,\ldots,T-1\}$, recall that  $S_i({\cal G})$ denote the 3-smashed graph $G_{i-1}\cup  G_i\cup G_{i+1}$. We prove that $\chi(S_i({\cal G}))\leq 6d$ (Lemma~\ref{lem:degenerate_ub}) and we show that it exists some graphs $G_1$, $G_2$ and $G_3$ that are all $d$-degenerate and need at least $5d$ colors to be proper colored (Lemma~\ref{lem:degenerate_lb}). Those results implies that the class of temporal graph whose snapshots are all $d$-degenerate have maximal chromatic number at least $5d$ and at most $12d$.

\begin{lemma}\label{lem:degenerate_ub}
If $G_1$, $G_2$ and $G_3$ are three $d$-degenerate graphs then $\chi(G_1\cup G_2\cup G_3)\leq 6d$.
\end{lemma}
%\begin{lemma}\label{lem:degenerate_ub}
%The chromatic number of the three-smashed graph of any three arbitrary $d$-degenerate snapshots is upper bounded by $6d$. 
%\end{lemma}

\begin{proof}
The number of edges in a $d$-degenerate graph on $n$ vertices is strictly less than $dn$. Hence the number of edges of $G_1\cup G_2\cup G_3$ is strictly less than $3dn$. Therefore, $G_1\cup G_2\cup G_3$ is a $(6d-1)$-degenerate graph and $G_1\cup G_2\cup G_3$ is $6d$ colorable. 
\end{proof}

\begin{lemma}\label{lem:degenerate_lb}
There exist three $d$-degenerate graphs, $G_1, G_2, G_3$, such that the 3-smashed graph $S_2$,  is a complete graph on $5d$ vertices.
\end{lemma}

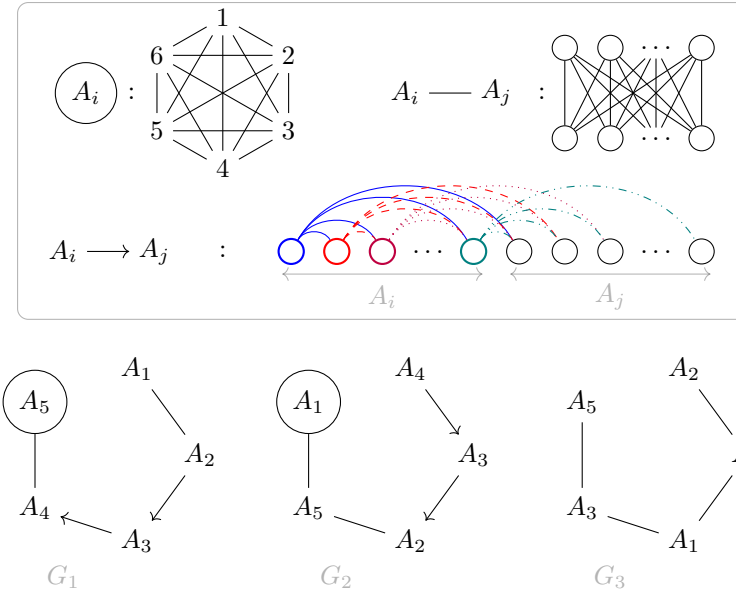
\begin{figure}
    \centering
    \centering

\begin{tikzpicture}[scale=0.6]
\begin{scope}[shift={(0,0)}]
\node at (360/5*1:2cm) (A1) {$A_1$};
\node[circle,draw] at (360/5*2:2cm) (A5) {$A_5$};
\node at (360/5*3:2cm) (A4) {$A_4$};
\node at (360/5*4:2cm) (A3) {$A_3$};
\node at (360/5*5:2cm) (A2) {$A_2$};
\draw (A1)--(A2);
\draw[->] (A2)--(A3);
\draw[->] (A3)--(A4);
\draw (A4)--(A5);
\end{scope}

\begin{scope}[shift={(6,0)}]
\node at (360/5*1:2cm) (A4) {$A_4$};
\node[circle,draw] at (360/5*2:2cm) (A1) {$A_1$};
\node at (360/5*3:2cm) (A5) {$A_5$};
\node at (360/5*4:2cm) (A2) {$A_2$};
\node at (360/5*5:2cm) (A3) {$A_3$};
\draw (A5)--(A2);
\draw[->] (A4)--(A3);
\draw[->] (A3)--(A2);
\draw (A1)--(A5);
\end{scope}

\begin{scope}[shift={(12,0)}]
\node at (360/5*1:2cm) (A2) {$A_2$};
\node at (360/5*2:2cm) (A5) {$A_5$};
\node at (360/5*3:2cm) (A3) {$A_3$};
\node at (360/5*4:2cm) (A1) {$A_1$};
\node at (360/5*5:2cm) (A4) {$A_4$};
\draw (A2)--(A4)--(A1)--(A3)--(A5);
\end{scope}

\begin{scope}[shift={(0,-2.7)}]
\draw[opacity=0.3] (-1,0) node[] ()   [] {$G_{1}$};
\draw[opacity=0.3] (5,0) node[] ()   [] {$G_{2}$};
\draw[opacity=0.3] (11,0) node[] ()   [] {$G_{3}$};
\end{scope}

\begin{scope}[shift={(-1.5,4.5)}]
\tikzstyle{whitenode}=[draw,circle,minimum size=7pt,inner sep=0pt]
\draw (0.5,0) node (ai)   [] {$A_i$};
\draw (2.5,0) node (aj)   [] {$A_j$};
\draw (4,0) node ()   [] {$:$};
\draw[->] (ai)--(aj);
\end{scope}

\begin{scope}[shift={(4,4.5)}]
\draw (0,0) node[circle, draw=blue, thick] (a0)   [] {};
\draw (1,0) node[circle,draw=red, thick] (a1)   [] {};
\draw (2,0) node[circle,draw=purple, thick] (a2)   [] {};
\draw (3,0) node[] (a3)   [] {$\ldots$};
\draw (4,0) node[circle,draw=teal, thick] (a4)   [] {};
\draw (5,0) node[circle,draw,black] (a5)   [] {};
\draw (6,0) node[circle,draw,black] (a6)   [] {};
\draw (7,0) node[circle,draw,black] (a7)   [] {};
\draw (8,0) node[] (a8)   [] {$\ldots$};
\draw (9,0) node[circle,draw,black] (a9)   [] {};
\draw[blue] (a0) edge [out=60,in=120] (a1);
\draw[blue] (a0) edge [out=60,in=120] (a2);
\draw[blue] (a0) edge [out=60,in=120] (a4);
\draw[blue] (a0) edge [out=60,in=120] (a5);
\draw[dashed,red] (a1) edge [out=60,in=120] (a2);
\draw[dashed,red] (a1) edge [out=60,in=120] (a4);
\draw[dashed,red] (a1) edge [out=60,in=120] (a5);
\draw[dashed,red] (a1) edge [out=60,in=120] (a6);
\draw[dotted,purple] (a2) edge [out=60,in=120] (a4);
\draw[dotted,purple] (a2) edge [out=60,in=120] (a5);
\draw[dotted,purple] (a2) edge [out=60,in=120] (a6);
\draw[dotted,purple] (a2) edge [out=60,in=120] (a7);
\draw[dashdotdotted,teal] (a4) edge [out=60,in=120] (a5);
\draw[dashdotdotted,teal] (a4) edge [out=60,in=120] (a6);
\draw[dashdotdotted,teal] (a4) edge [out=60,in=120] (a7);
\draw[dashdotdotted,teal] (a4) edge [out=60,in=120] (a9);

\draw[<-> , opacity=0.3] (-0.2,-0.5)--(4.2,-0.5);
\draw[opacity=0.3] (2,-1 )node  () [] {$A_i$};
\draw[<-> , opacity=0.3] (4.8,-0.5)--(9.2,-0.5);
\draw[opacity=0.3] (7,-1 )node  () [] {$A_j$};
\end{scope}

\begin{scope}[shift={(-0.5,8)}]
   \node[circle,draw] at (0,0) (ai) {$A_i$};
   \node at (1,0) () {$:$};
\end{scope}
\begin{scope}[shift={(2.5,8)}]
    \graph{ subgraph K_n [n=6, clockwise ]};
\end{scope}

\begin{scope}[shift={(6,8)}]
\tikzstyle{whitenode}=[draw,circle,minimum size=7pt,inner sep=0pt]
\draw (0.5,0) node (ai)   [] {$A_i$};
\draw (2.5,0) node (aj)   [] {$A_j$};
\draw (3.5,0) node ()   [] {$:$};
\draw (ai)--(aj);
\end{scope}
\begin{scope}[shift={(10,7)}]
    \draw (0,0) node[circle, draw] (a0)   [] {};
    \draw (1,0) node[circle,draw] (a1)   [] {};
    \draw (2,0) node[] (a)   [] {$\ldots$};
    \draw (3,0) node[circle,draw] (a2)   [] {};
    \draw (0,2) node[circle,draw] (a4)   [] {};
    \draw (1,2) node[circle,draw] (a5)   [] {};
    \draw (2,2) node[] (b)   [] {$\ldots$};
    \draw (3,2) node[circle,draw] (a6)   [] {};
    \graph [] {
    {(a0),(a1),(a2),(a)} -- [complete bipartite] {(a4),(a5),(a6),(b)} };
\end{scope}

\begin{scope}[shift={(-2,3)}]
    \draw[opacity=0.3,rounded corners] (0,0) rectangle (16,7);
\end{scope}

\end{tikzpicture}
    \caption{Three consecutive $d$-degenerate snapshots $G_{1}, G_2,G_3$ such that the $3$-smashed graph $G_1 \cup G_2 \cup G_3$ is a complete graph. For every $i \in \{1,\ldots,5\}$, each set $A_i$ includes $d$ vertices. A circle around set $A_i$ means that we connect any two vertices of $A_i$ to make a clique, $A_i-A_j$ means we connect every vertex of $A_i$ with every vertex of $A_j$ to make a complete bipartite graph, and $A_i\rightarrow A_j$ means that we order vertices by $A_i,A_j$, then every vertex in $A_i$ are connected to the next $d$ vertices with respect to ordering $A_i,A_j$. The degenerate orderings of $G_1, G_2,G_3$ are respectively $(A_1 \leq A_2 \leq A_3 \leq A_4 \leq A_5)$, $(A_4 \leq A_3 \leq A_2 \leq A_5 \leq A_1)$, and $(A_2 \leq A_4 \leq A_1 \leq A_3 \leq A_5)$.}
    \label{fig:degenerate}
\end{figure}

\begin{proof} 
Let $A_1,\ldots, A_{5}$ be five sets of $d$ vertices each and set $V=A_1\cup A_2\cup A_3\cup A_4\cup A_5$ (so $|V|=5d$). For every $i \in [5]$, define a linear ordering on vertices of $A_i$, that is denoted by $a_i^1\leq \ldots\leq a_i^{d}$. For every $i_1,\ldots,i_{\ell} \in [5]$, we use notation $(A_{i_1}\leq \ldots \leq A_{i_{\ell}})$ to represent the set of linear-ordered vertices induced by the ordering of vertices in each bag, i.e.  $a_{i_1}^1\leq \ldots\leq a_{i_1}^d\leq \ldots\leq a_{i_{\ell}}^1\leq \ldots\leq a_{i_{\ell}}^d$.

An illustration of the construction can be found in Figure~\Ref{fig:degenerate}. We define  three following types of edge sets:

\begin{itemize}
    \item For every two bags $A_i$ and $A_j$, let $T_1(A_i,A_j)$ be the edge set of a complete bipartite graph between $A_i$ and $A_j$ (denoted by $A_i-A_j$ in Figure~\ref{fig:degenerate}). Hence,  $T_1(A_i,A_j)=\{a_ia_j \mid a_i\in A_i \mbox{ and }a_j\in A_j\}$.
    \item For every two bags $A_i$ and $A_j$, let $T_2(A_i,A_j)$ be the set of edges connecting every vertex in $A_i$ with the next $d$ consecutive vertices with respect to the ordering $a_i^1,\ldots,a_i^d,a_j^1,\ldots,a_j^d$ (denoted by $A_i\rightarrow A_j$ in Figure~\ref{fig:degenerate}). Hence, $T_2(A_i,A_j)=T_3(A_1)\cup \{a^k_i a^l_j \mid a^k_i\in A_i, a^l_j\in A_j \mbox{ and }l\leq k\} $.
    \item For every bag $A_i$, let $T_3(A_i)$ be a set of edges connecting all pairs of vertices in $A_i$, that make $A_i$ be a clique (denoted as a circled $A_i$ in Figure~\ref{fig:degenerate}). Hence, $T_3(A_i)=\{a_ia_i' \mid a_i,a_i'\in A_i \}$.
\end{itemize}

Observe that the union of  $T_2(A_i,A_j)\cup T_2(A_j,A_i)$ induces a complete graph on $A_i\cup A_j$. Indeed, we already remark that it implies $A_i$ and $A_j$ are cliques. In addition, for all $a_i^k\in A_i$ and $a_j^l\in A_j$ either $k\leq l$ and $a_i^k a_j^l \in T_2(A_j,A_i)$ or  $l\leq k$ and $a_i^k a_j^l \in T_2(A_j,A_i)$. 

\vspace{1ex}

Now, we are ready to construct the three snapshots along with their degenerate orderings. Observe that the degenerate ordering is not the same for all 3 snapshots.

\begin{enumerate}
    \item Snapshot $G_1 = (V,E_1)$: $E_1 = T_1(A_1,A_2) \cup T_1(A_4,A_5) \cup T_2(A_2,A_3) \cup T_2(A_3,A_4) \cup T_3(A_5)$, and the ordering of vertices is $A_1\leq A_2\leq A_3\leq A_4\leq A_5$;
    \item Snapshot $G_2 = (V,E_2)$: $E_2 = T_1(A_2,A_5) \cup T_1(A_5,A_1) \cup T_2(A_4,A_3) \cup T_2(A_3,A_2) \cup T_3(A_1)$, the ordering of vertices is $A_4\leq A_3\leq A_2\leq A_5\leq A_1$;
    \item Snapshot $G_3 = (V,E_3)$: $E_3 = T_1(A_2,A_4) \cup T_1(A_4,A_1), T_1(A_1,A_3) \cup T_1(A_3,A_5)$, the ordering of vertices is $A_2\leq A_4\leq A_1\leq A_3\leq A_5$.
\end{enumerate}

For a set of linear-ordered vertices $x_1,\ldots,x_{\ell}$, for every index $j$, let the set of vertices to the left of $x_j$ be $\{ x_{j+1},\ldots,x_{\ell} \}$.

We prove that graphs $G_1,G_2,G_3$ are $d$-degenerate graphs with the defined vertex orderings.  Consider snapshot $G_1$ with the vertex ordering $A_1,A_2,A_3,A_4,A_5$ we have: a vertex in $A_1$ has exactly $d$ neighbors to its left that are in $A_2$, a vertex in $A_2$ has $d$ neighbors to its left that are in $A_2$ and $A_3$, a vertex in $A_3$ has $d$ neighbors to its left that are in $A_3$ and $A_4$, a vertex in $A_4$ has $d$ neighbors to its left that are in $A_5$, and a vertex in $A_5$ has $d$ neighbors to its left that are in $A_5$. It implies that $G_1$ is a $d$-degenerate graph. The proof on $G_2$ and $G_3$ being similar, we omit it.

\vspace{1.5ex}

We now prove that the three-smashed graph $G_1 \cup G_2 \cup G_3$ is a complete graph on $5d$ vertices. We prove that for all $i,j\in [5]$, $A_i$ is a clique and $A_i$ is complete to $A_j$.

\begin{itemize}
    \item For all $i\in [5]$, $A_i$ is a clique:  $A_1$ is a clique $G_2$, $A_2$ is a clique $G_1$ (as $T_2(A_2,A_3)$ belongs to $G_1$), $A_3$ and $A_4$ is connected in $G_2$ (as $T_2(A_3,A_2)$ and $T_2(A_4,A_3)$ belongs to $G_2$),  and  $A_5$ is a clique in $G_1$.
    \item For all $i\in \{2,3,4,5\}$, $A_1$ is complete to $A_i$: $A_1$ is complete to $A_2$ and $A_5$ in $G_1$ and $A_1$ is complete to $A_3$ and $A_4$ in $G_3$.  
    \item For all $i\in \{3,4,5\}$, $A_2$ is complete to v $A_i$: $A_2$ is complete to $A_3$  as $T_2(A_2,A_3)\cup T_2(A_3,A_2)$ belongs to  in $G_1\cup G_2$, $A_2$ is complete to $A_4$ in $G_3$ and $A_2$ is complete to $A_5$ in $G_2$.  
    \item For $i \in \{4,5\}$, $A_3$ is complete to $A_i$. $A_3$ is complete to $A_4$ by $T_2(A_3,A_4)$ in $G_1$ and $T_2(A_4,A_3)$ in $G_2$. $A_3$ is complete to $A_5$ in $G_3$.
    \item $A_4$ is complete to  $A_5$ in $G_1$.
\end{itemize}

Hence $G_1 \cup G_2 \cup G_3$ is a complete graph on $5d$ vertices and that concludes the proof. 
\end{proof}

Theorem~\ref{thm:degenerate} is a direct consequence of Theorem~\ref{thm:ub2}, Lemma~\ref{lem:degenerate_ub} and Lemma~\ref{lem:degenerate_lb}.

\begin{theorem}\label{thm:degenerate}
    Let $\mathcal{G}$ be a temporal graph whose snapshots are $d$-degenerate graphs. The temporal chromatic number of $\mathcal{G}$ is between $5d$ and $12d$, i.e., $5d \leq \chi^t(\mathcal{G}) \leq 12d$.
\end{theorem}

%%%%%%%%%%%%%%%%%%%%%%%%%%%%%%%%%%%%%%%%%%%%%%%%%%%%%%%%%%%%%%%
\subsection{Bounded-Degree Graphs}\label{degbounded}
%%%%%%%%%%%%%%%%%%%%%%%%%%%%%%%%%%%%%%%%%%%%%%%%%%%%%%%%%%%%%%%

In this section, we give a tight upper bound of $(3\Delta+1)$ on the chromatic number of  3-smashed graphs with snapshots being all  $\Delta$-bounded graphs. Theorem~\Ref{thm:ub2} implies that there exists a $(6\Delta+2)$-coloring for every temporal graph with all snapshots being $\Delta$-bounded graph. In this section, we improve this bound by showing that  the chromatic number of a temporal graph with all snaphots being $\Delta$-bounded graphs is at most  $(5\Delta+1)$.

\begin{lemma}\label{lem:boundedDeg_lb}
The chromatic number of the three-smashed graph of any three arbitrary $\Delta$-bounded graphs is at most $(3\Delta+1)$. Moreover, this bound is tight.
\end{lemma}

\begin{proof}
Let $G_{i-1}, G_i, G_{i+1}$ be three $\Delta$-bounded graphs. The maximum degree of $G_{i-1} \cup G_i \cup G_{i+1}$ is at most $3\Delta$, so the 3-smashed graph $G_{i-1}\cup G_i\cup  G_{i+1}$ is $(3 \Delta +1)$ -colorable.

The last part of the proof define three snapshots to show that the above upper bound is tight. 

Let $V= \{u\} \cup W_1 \cup W_2 \cup W_3$ be a vertex set such that, for $i\in [3]$, $|S_i|=\Delta$. For $i\in [3]$, let $G_i = (V,E_i)$, be three $\Delta$-bounded graphs whose edge sets is

\begin{itemize}
    \item $E_1  = \{uv \mid u,v \in W_1 \cup \{u\}\} \cup \{uv \mid u\in W_2 \mbox{ and } v \in W_3 \}$; 
    \item $E_2 = \{uv \mid u,v \in W_2 \cup \{u\} \} \cup \{uv \mid u \in W_1 \mbox{ and } v \in W_3 \}$;
    \item $E_3 = \{uv \mid u,v \in W_3 \cup \{u\} \} \cup \{uv \mid u \in W_2 \mbox{ and } v \in W_1 \}$.
\end{itemize}

Observe that the graph $G_1 \cup G_2 \cup G_3$ is a complete graph on $(3\Delta+1)$ vertices, so its chromatic number is $(3\Delta +1)$.
\end{proof}

\begin{lemma}\label{lem:boundedDeg_ub}
Let ${\cal G}=(G_1,\dots, G_T)$ be a temporal graph whose snapshots $G_i$ is a $\Delta$-bounded graph, for $i\in [T]$. The chromatic number of $\mathcal{G}$ is at most $(5\Delta +1)$, i.e., $\chi^t({\cal G})\leq 5\Delta +1$.
\end{lemma}

\begin{proof}
We prove by induction  on $i\le T$ that colorings $c_1,\ldots,c_i$ can be computed. We compute $c_1$ as a $(2\Delta+1)$-coloring of the $2\Delta$-bounded degree smashing $G_1\cup G_2$. Assume that there is a $(5\Delta+1)$-coloring $c_1,\ldots,c_{i}$ satisfying that:
    \begin{itemize}
        \item For every $j \in [i]$, $c_j$ is a proper coloring of $G_{j-1} \cup G_{j} \cup G_{j+1}$.
        \item For every $j \in [ i-1]$,  $c_j$ and $c_{j+1}$ are compatible on $G_j \cup G_{j+1}$.
    \end{itemize}

We show that there exists $c_{i+1}$ such that (1) $c_{i+1}$ is a proper coloring of $G_i \cup G_{i+1} \cup G_{i+2}$; (2) $c_{i+1}$ and $c_i$ are compatible on $G_i \cup G_{i+1}$. In every two-smashed graph $G_i \cup G_{i+1}$, each vertex $v$ has at most $2 \Delta$ neighbors. Therefore, there are at least $(3\Delta+1)$ colors that can be used to color $v$ in $c_{i+1}$ without violating the compatibility condition. In other words, for every vertex $v$, there are a list $L_v$ of $(3\Delta+1)$ colors such that if every vertex $v$ choose a color $c_{i+1}(v) \in L_v$, then $c_{i+1}$ is compatible with $c_i$. It is left to prove that there is a way to choose $c_{i+1}(v) \in L_v$ for every vertex $v$ such that $c_{i+1}$ is a proper coloring for $G_i \cup G_{i+1} \cup G_{i+2}$. 
    
Observation that a $d$-degenerate graph $D$ is $(d+1)$-list colorable. Indeed, let $x_1,\ldots,x_s$ be a degeneracy order of $D$. We pick colors for vertices of $D$ by the reverse ordering $x_s,\ldots,x_1$. For every $i \in \{1,\ldots,s\}$, $x_i$ have at most $d$ neighbors, in $\{x_s,\ldots,x_{i+1}\}$, which are already colored. Thus, there is an available color for $x_i$ from its color list of size $(d+1)$.
    
The 3-smashed graph $G_i \cup G_{i+1}\cup G_{i+2}$ is $3 \Delta$-bounded graph, so it is $3\Delta$-degenerate graph. Therefore, $G_i \cup G_{i+1}\cup G_{i+2}$ is $(3\Delta+1)$-list colorable. It implies that there is a coloring $c_{i+1}$ satisfying the two mentioned conditions.
\end{proof}

Now, Theorem~\ref{thm:boundedDeg} follows directly from Lemma~\ref{lem:boundedDeg_lb} and Lemma~\ref{lem:boundedDeg_ub}.

\begin{theorem}\label{thm:boundedDeg}
    The maximal chromatic number of temporal graphs where each snapshot is $\Delta$-bounded is between $3\Delta+1$ and $5\Delta+1$.
\end{theorem}

As in Theorem~\ref{th:copy}, by duplicating each snapshot, we get the following upper bound for the temporal chromatic number:

\begin{theorem}
Let $\mathcal{G} = (G_1$, $G_2$, \ldots, $G_T)$ be a temporal graph with $T$ even in which $G_{2i-1}=G_{2i}$ for all $i\in[T/2]$ (i.e. there are two consecutive copies for each snapshot). If for all $i \in [T]$, $G_i$ has bounded degree $\Delta$, then $\chi^t(\mathcal{G})\le 3\Delta +1$.
\end{theorem}
\begin{proof}
We use the same notations as in the proof of Theorem~\ref{th:copy}. More precisely, we define $H_i=G_{2i-1}=G_{2i}$ and see the temporal graph $\mathcal{G}$ as $(H_1$, $H_1$, $H_2$, $H_2$, $\ldots$, $H_{T/2}$, $H_{T/2})$.

We build the colorings sequentially as follows:

\begin{enumerate}
    \item  Choose $c_1=c_2=c_3$ as a proper $(3\Delta+1)$-coloring of the $2\Delta$-bounded graph $H_1 \cup H_2$. 
    \item For $i>1$, we compute sequentially and greedily $c_{i}$ on each vertex one after another. Briefly, if $i = (2j+1)$ is odd, then choose $c_{2j+1}=c_{2j}$. If $i= 2j$ is even, then we choose $c_{2j}$ such that $c_{2j}$ a proper coloring of $H_j \cup H_{j+1}$, compatible with $c_{2j-1}$ on $H_{j}$. Here is how we build $c_{2j}$ from $c_{2j-1}$:
    
    For every vertex $v$, to compute $c_{2j}(v)$, avoid all the colors in $\{c_{2j-1}(u)|u\in N_{H_j}(v)\}$. % where $u$ is a neighbour of $v$ in $H_j$.
    This constraint ensures that $c_{2j}$ is compatible with $c_{2j-1}$. Since every vertex $v$ has maximum degree $\Delta$ in $H_j$, it follows that, for every vertex $v$,  up to $\Delta$ colors needs to be avoided when selecting $c_{2j}(v)$. In addition, since we compute the coloring $c_{2j}$ greedily, $c_{2j}(v)$ needs to be different from $c_{2j}(u)$ for every colored neighbors $u$ of $v$ in $H_j \cup H_{j+1}$. Every vertex $v$ has at most $2\Delta$ neighbors in $H_j \cup H_{j+1}$, so we need to avoid up to $2\Delta$ more colors for $c_{2j}(v)$. 
    As we have $(3\Delta+1)$ colors, it is always possible to choose a color for $v$ respecting the two mentioned constrains.

\end{enumerate}
By construction, $(c_1,\dots, c_{i+1})$ is a temporal $3\Delta+1$-coloring of $(G_1,\dots G_{i+1})$. Then, by induction, $(c_1,\dots, c_{T})$ is a temporal $3\Delta+1$-coloring of ${\cal G}$. %Note that it is computed online.
\end{proof}

\section{Coloring temporal degree bounded graphs with grow pace one}\label{sec:grow1}

We consider now the case where the evolution of the set of edges between consecutive snapshots is the simplest: at most one edge is added and at most another one is removed at each time step. We show in this section coloring results depending on the maximal degree $\Delta$ of the graph. Since the complete graph of $\Delta+1$ vertices has maximal degree $\Delta$, we know that $\Delta+1$ colors could be necessary. Figure~\ref{fig:chromatic4} provides a temporal graph with $\Delta=2$ and temporal chromatic number 4.

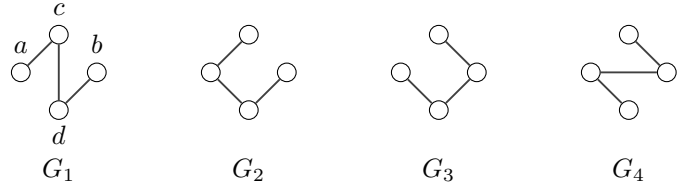
\begin{figure}[ht]
        \centering
    \begin{tikzpicture}
    \tikzstyle{whitenode}=[draw,circle,minimum size=7pt,inner sep=0pt]
    %\Vertex[x=-4,y=0, label=$a$, position=180]{a1}
    % \Vertex[x=-3,y=0, label=$b$, position=0]{b1}
    % \Vertex[x=-3.5,y=.5, label=$c$, position=90]{c1}
    % \Vertex[x=-3.5,y=-.5, label=$d$, position=below]{d1}
    
    \draw (-4,0) node[whitenode,  label=$a$] (a1)   [] {};
    \draw (-3,0) node[whitenode,  label=$b$] (b1)   [] {};
    \draw (-3.5,.5) node[whitenode,  label=$c$] (c1)   [] {};
    \draw (-3.5,-.5) node[whitenode,  label={[shift={(0,-.7)}]$d$}] (d1)   [] {};

    \node[] at (-3.5,-1.3) {$G_1$};
    % \Vertex[x=-1.5,y=0]{a2}
    % \Vertex[x=-.5,y=0]{b2}
    % \Vertex[x=-1,y=.5]{c2}
    % \Vertex[x=-1,y=-.5]{d2}
    
    \draw (-1.5,0) node[whitenode] (a2)   [] {};
    \draw (-.5,0) node[whitenode] (b2)   [] {};
    \draw (-1,.5) node[whitenode] (c2)   [] {};
    \draw (-1,-.5) node[whitenode] (d2)   [] {};
    \node[] at (-1,-1.3) {$G_2$};
    % \Vertex[x=1,y=0]{a3}
    % \Vertex[x=2,y=0]{b3}
    % \Vertex[x=1.5,y=.5]{c3}
    % \Vertex[x=1.5,y=-.5]{d3}
    
    \draw (1,0) node[whitenode] (a3)   [] {};
    \draw (2,0) node[whitenode] (b3)   [] {};
    \draw (1.5,.5) node[whitenode] (c3)   [] {};
    \draw (1.5,-.5) node[whitenode] (d3)   [] {};
    
    \node[] at (1.5,-1.3) {$G_3$};

    % \Vertex[x=3.5,y=0]{a4}
    % \Vertex[x=4.5,y=0]{b4}
    % \Vertex[x=4,y=.5]{c4}
    % \Vertex[x=4,y=-.5]{d4}
    
    \draw (3.5,0) node[whitenode] (a4)   [] {};
    \draw (4.5,0) node[whitenode] (b4)   [] {};
    \draw (4,.5) node[whitenode] (c4)   [] {};
    \draw (4,-.5) node[whitenode] (d4)   [] {};
    
    \node[] at (4,-1.3) {$G_4$};

    \Edge(a1)(c1);
    \Edge(c1)(d1);
    \Edge(d1)(b1);

    \Edge(c2)(a2);
    \Edge(a2)(d2);
    \Edge(d2)(b2);

    \Edge(a3)(d3);
    \Edge(d3)(b3);
    \Edge(b3)(c3);
    
    \Edge(c4)(b4);
    \Edge(b4)(a4);
    \Edge(a4)(d4);
\end{tikzpicture}
    \caption{Example of  temporal graph with {grow pace}~1, each snapshot being a path, with temporal chromatic number~4. 
    }
    \label{fig:chromatic4}
\end{figure}

We first show in Theorem~\ref{thm:delta_grow1} that $\Delta+2$ colors are always enough, and then discuss whether they are necessary or not, depending on the value of $\Delta$.  %Moreover, our construction works online: we compute $c_{i+1}$ with only the knowledge of any $c_i$ and the snapshots up to $G_{i+2}$. 

%Theorem~\ref{thm:delta_grow1} that any temporal graphs of grow pace 1 where each snapshot is $\Delta$-bounded can be temporally $(\Delta+2)$-colored.

\begin{theorem}\label{thm:delta_grow1}
    Let $\mathcal{G} = (G_1,\ldots,G_T)$ be a temporal graph with grow pace 1. If for all $i\in [T]$, $G_i$ is $\Delta$ bounded, then $\chi^t({\mathcal{G}})\leq \Delta+2$.
\end{theorem}

\begin{proof}
    We prove by induction on $i\le T$ that colorings $c_1,\ldots,c_i$ can be computed for $\mathcal{G}$. To compute $c_1$, we need a coloring of $S_1=G_1\cup G_2$. In $S_1$, at most two vertices have degree bigger than $\Delta$ (those two vertices being the extremities of the edge added in $G_2$). Give colors 1 and 2 to those two nodes, and we can greedily color the other nodes with $(\Delta+1)$ colors.
    
    Assume that there is a $(\Delta+2)$-coloring $c_1,\ldots,c_{i}$ satisfying:
    \begin{itemize}
        \item For every $j \in [i]$, $c_j$ is a proper coloring of $G_{j-1} \cup G_{j} \cup G_{j+1}$.
        \item For every $j \in [ i-1]$,  $c_j$ and $c_{j+1}$ are compatible on $G_j \cup G_{j+1}$.
    \end{itemize}
    We show that there exists $c_{i+1}$ such that $c_{i+1}$ is a proper coloring of $G_i \cup G_{i+1} \cup G_{i+2}$; and $c_{i+1}$ and $c_i$ are compatible on $G_i \cup G_{i+1}$.
    
    For all vertices $u$ that are not an endpoint of the edge added in $G_{i+2}$, set $c_i(u)=c_{i+1}(u)$.
    
    Suppose that there exists an edge $uv \in E_{i+2} \setminus E_{i+1}$. Since the temporal graph has \rate~1, this edge is unique.  If $c_i(u) \neq c_i(v)$, then let $c_{i+1}(u) = c_{i}(u)$ and  $c_{i+1}(v) = c_{i}(v)$. Suppose that $c_i(u) = c_i(v)$. Observe that $G_{i+2}$ is $\Delta$ bounded. 
    
    Since the \rate~is $1$  there exists at most one unique edge $e$ in $G_{i+1}$ that does not exists in $G_{i+2}$. This edge cannot be $uv$, as $uv\in G_{i+2}$. Assume that both $u$ and $v$ have degree $\Delta$ in $G_{i+1}$. The removal of this edge $e$ can decrease the degree of either $u$ or $v$, but not both, and the addition of $uv$ make the degree of either $u$ or $v$ become $(\Delta+1)$ in $G_{i+2}$.
    Therefore, either $u$ or $v$ has degree at most $(\Delta-1)$ in $G_{i+1}$. 
    
    It implies that, in $G_i \cup G_{i+1}$, either $u$ or $v$ has degree at most $\Delta$. Without loss of generality, suppose that $u$ has degree at most $\Delta$ in $G_i \cup G_{i+1}$ and set $c_{i+1}(v)=c_i(v)$. Since we use $(\Delta +2)$ colors, there exists a color $c_{i+1}(u) \neq c_{i}(u)$, and for every neighbor $w$ of $u$ in $G_i \cup G_{i+1}$, $c_{i+1}(u) \neq c_i(w)$. In particular $c_{i+1}(u)\neq c_{i+1}(v)$.
    
    The coloring $c_{i+1}$ is proper for $G_i \cup G_{i+1} \cup G_{i+2}$: 
    
    \begin{itemize}
        \item For all $xy\in E(G_i \cup G_{i+1} \cup G_{i+2})$, either $xy\in E(G_i\cup G_{i+1})$ and by induction hypothesis, $c_{i+1}(x)\neq c_{i+1}(y)$, or $xy=uv$ and, the same results holds by previous observation.
        \item For all edges $xy\in E(G_i \cup G_{i+1})$ then either $c_{i}(x)= c_{i+1}(x)$ and $c_{i}(y)= c_{i+1}(y)$, $u=x$ or $u=y$. In the first case, by induction hypothesis, $c_{i}(x)\neq c_{i+1}(y)$. In the two other cases, $c_{i+1}(u)$ differs from the colors of its neighbors in both $c_i$ and $c_{i+1}$. Hence coloring $c_{i+1}$ is compatible with coloring $c_i$.
    \end{itemize} 
    
 By induction, $(c_1,\dots, c_{T})$ is a temporal $\Delta+2$-coloring of ${\cal G}$. %Note that it is computed online. 
 \end{proof}

For the cases $\Delta=2$ and 3, we provide temporal graphs that have temporal chromatic number $(\Delta+2)$. Those temporal graphs were found using a program that computes all possible temporal graphs on $(\Delta+2)$ nodes with grow pace 1 with a short lifetime. We ran the program for $\Delta=4$, $n=6$ and $T=4$ (which suffices to ensure that all edges appear at least on one snapshot), and did not find any temporal graph that needs $6$ colors.

\begin{proposition}\label{lemma:P4}
There exists a temporal graph $\cal G$ with \rate~1 and each snapshot being a path such that $\chi^t({\cal G})= 4$.
\end{proposition}
\begin{proof}
    The temporal graph represented in Figure~\ref{fig:chromatic4} has \rate~1 and each snapshot is a path. We prove that its temporal chromatic number is~$4$.
    
    Suppose it exists a temporal 3-coloring $(c_1,c_2,c_3,c_4)$ on color set $\{1,2,3\}$ for the temporal graph represented in Figure~\ref{fig:chromatic4}. In $G_1 \cup G_2 \cup G_3$, observe that the vertices $b$, $c$, and $d$ form a triangle. Without loss of generality, we can suppose that $c_2(b)=1$, $c_2(c)=2$, and $c_2(d)=3$. Since $a$ is a neighbor of $c$ and $d$ in $G_2$, we have $c_2(a)=1$. Now, observe that $ab$ is an edge in $G_4$. Hence $c_3(a)\neq c_3(b)$ and, since there are only colors $\{1,2,3\}$, one of them belongs to $\{2,3\}$. Therefore, either $c_3(b) \in \{c_2(c),c_2(d)\}$ or $c_3(a) \in \{c_2(c),c_2(d)\}$. This is a contradiction as the colorings $c_2$ and $c_3$ need to be compatible in $G_2\cup G_3$, and in such a graph, both $a$ and $b$ are neighbors of $c$ and $d$.
\end{proof}

\begin{proposition}

There exists a temporal graph $\cal G$ with \rate~1 where each snapshot has maximal degree $\Delta=3$ such that $\chi^t({\cal G})= 5$.
\end{proposition}

\begin{proof}
\begin{figure}[!ht]
    \centering
\begin{tikzpicture}
   
    \tikzstyle{whitenode}=[draw,circle,fill=white,minimum size=7pt,inner sep=0pt]
            %\Vertex[x=1,y=1.25, position=above, label=$b$]{v11}
        %\Vertex[x=0,y=.75, position=180, label=$a$]{v21}
       % \Vertex[x=.5,y=0, position=180,label=$e$]{v31}
       % \Vertex[x=2,y=.75, position=0, label=$c$]{v41}
        %\Vertex[x=1.5,y=0, position=0, label=$d$]{v51}

     \draw (0,0) node[whitenode] (e1) [label=180:$e$] {}
-- ++(0:1cm) node[whitenode] (d1) [ label=0:$d$] {}
-- ++(72:1cm) node[whitenode] (c1) [label=0:$c$] {}
-- ++(144:1cm) node[whitenode] (b1) [label=$b$] {}
-- ++(216:1cm) node[whitenode] (a1) [label=180:$a$] {};

\draw (a1) edge  node {} (c1);
\draw (a1) edge  node {} (e1);
\draw (b1) edge  node {} (e1);

    \node[] at (.5,-.5) {$G_1$};
    
     \draw (2.75,0) node[whitenode] (e2) {}
-- ++(0:1cm) node[whitenode] (d2) {}
-- ++(72:1cm) node[whitenode] (c2) {}
-- ++(144:1cm) node[whitenode] (b2) {}
-- ++(216:1cm) node[whitenode] (a2) {};

\draw (a2) edge  node {} (e2);
\draw (a2) edge  node {} (d2);
\draw (b2) edge  node {} (e2);

    \node[] at (3.25,-.5) {$G_2$};
    
     \draw (5.5,0) node[whitenode] (e3) {}
 ++(0:1cm) node[whitenode] (d3) {}
-- ++(72:1cm) node[whitenode] (c3) {}
-- ++(144:1cm) node[whitenode] (b3) {}
-- ++(216:1cm) node[whitenode] (a3) {};

\draw (a3) edge  node {} (e3);
\draw (a3) edge  node {} (d3);
\draw (b3) edge  node {} (e3);
\draw (c3) edge  node {} (e3);

    \node[] at (6,-.5) {$G_3$};
    
     \draw (8.25,0) node[whitenode] (e4) {}
 ++(0:1cm) node[whitenode] (d4) {}
-- ++(72:1cm) node[whitenode] (c4) {}
 ++(144:1cm) node[whitenode] (b4) {}
-- ++(216:1cm) node[whitenode] (a4) {};

\draw (a4) edge  node {} (d4);
\draw (b4) edge  node {} (e4);
\draw (c4) edge  node {} (e4);
\draw (b4) edge  node {} (d4);
\draw (a4) edge  node {} (e4);

    \node[] at (8.75,-.5) {$G_4$};
    
     \draw (2.75,2.75) node[whitenode] (e5) {}
-- ++(0:1cm) node[whitenode] (d5) {}
-- ++(72:1cm) node[whitenode] (c5) {}
-- ++(144:1cm) node[whitenode] (b5) {}
-- ++(216:1cm) node[whitenode] (a5) {};

\draw (a5) edge  node {} (e5);
\draw (a5) edge  node {} (c5);
\draw (a5) edge  node {} (d5);
\draw (b5) edge  node {} (e5);
\draw (c5) edge  node {} (e5);

    \node[] at (3.25,2.25) {$S_2(\mathcal{G})$};
    
     \draw (5.5,2.75) node[whitenode] (e6) {}
-- ++(0:1cm) node[whitenode] (d6) {}
-- ++(72:1cm) node[whitenode] (c6) {}
-- ++(144:1cm) node[whitenode] (b6) {}
-- ++(216:1cm) node[whitenode] (a6) {};

\draw (a6) edge  node {} (e6);
\draw (a6) edge  node {} (d6);
\draw (b6) edge  node {} (e6);
\draw (c6) edge  node {} (e6);
\draw (b6) edge  node {} (d6);

    \node[] at (6,2.25) {$S_3(\mathcal{G})$};

    \end{tikzpicture}
    
    \caption{A temporal graph with temporal chromatic number~5, consisting of snapshots where each individual snapshot has a maximum degree of~$3$.  The graphs above correspond to the 3-smashed graphs. Only the first snapshot is explicitly labeled.}
    \label{fig:Delta3}
    \end{figure}
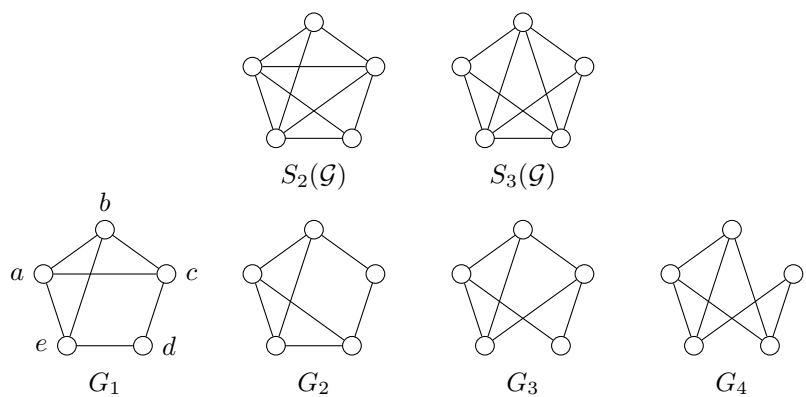

     The temporal graph represented in Figure~\ref{fig:Delta3} has \rate~1 and each snapshot has maximum degree 3. We prove that its temporal chromatic number is~$4$.
     
    Suppose it exits a temporal 4-coloring $(c_1,c_2,c_3,c_4)$ on color set $\{1,2,3,4\}$ for the temporal graph represented in Figure~\ref{fig:Delta3}.  Observe that $S_2(\mathcal{G})$ corresponds to the complete graph on 5 vertices without the edge $bd$. Hence, without loss of generality, we can assume $c_2(e)=1$, $c_2(b)=c_2(d)=2$, $c_2(a)=3$ and $c_2(c)=4$.

    Using compatibility and the edges of the second snapshot, we have $c_3(b) \notin \{c_2(e),c_2(a),c_2(c)\}$. In the same way, we get $c_3(d) \notin \{c_2(e),c_2(a),c_2(c)\}$. As we are using only~$4$ colors we get that $c_3(d)=c_3(b)$. This cannot be possible, since this edge is added in $G_4$.
\end{proof}

For $\Delta\le3$, we have temporal graphs with grow pace 1 with temporal chromatic number $\Delta+2$ colors. However, we believe that for $\Delta$ large enough (and maybe even already $\Delta=4$), $\Delta+1$ colors are always enough. Here is our intuition:

Let $\mathcal{G}$ be a temporal graph, and let $G$ denote the graph obtained by taking the smash of all snapshots of $\mathcal{G}$. Clearly, $\chi(G) \geq \chi^t(\mathcal{G})$. Consequently, if $\chi^t(\mathcal{G}) = \Delta + 2$, then $G$ has a chromatic number of at least $\Delta + 2$. This implies that $G$ contains at least $\Delta + 2$ vertices of degree at least $\Delta + 1$. Now, suppose $G_1$ is $\Delta$-bounded. To accommodate this number of high-degree vertices in $G$, we must add at least $\frac{\Delta}{2} + 1$ edges. However, given that $\mathcal{G}$ has \rate~1, this implies that $\mathcal{G}$ must have a lifetime of at least $\frac{\Delta}{2} + 1$. With a lifetime large enough, we do not think we can build a temporal graph with temporal chromatic number $\Delta+2$.

%%%%%%%%%%%%%%%%%%%%%%%%%%%%%%%%%%%%%%%%%%%%%%%%%%%%%%%%%%%%%%%
\section{Open Questions}\label{sec:conc}
%%%%%%%%%%%%%%%%%%%%%%%%%%%%%%%%%%%%%%%%%%%%%%%%%%%%%%%%%%%%%%%

In this paper, we have introduced a notion of coloring compatibility between consecutive snapshots in temporal graphs. The subsequent notion of temporal chromatic number leads to several open problems in graph theory. In particular, restricting the class of graphs for each snapshot or the number of changes between two consecutive snapshots open numerous questions.

Here are some questions that we think are of great interest for future work:
\begin{itemize}
    \item If each snapshot is a tree, does there exist a temporal graph needing 7 or 8 colors?
    \item If each snapshot is $\Delta$ bounded and the grow pace is 1, does there exist a temporal graph that needs $\Delta+2$ colors for some or all $\Delta\ge4$?
    \item For temporal graphs with grow pace 1 where each snapshot is bipartite, can we temporally 4-color them? %If so, does it also work online?
    \item What other constraints can be considered? For example, what can we say about the temporal chromatic number if we limit the number of neighbors that are changed for each vertex?
    \item Our construction in Theorem~\ref{thm:delta_grow1} starts from an arbitrary coloring $c_i$ to compute $c_{i+1}$ with $\Delta+2$ colors. On the other hand, Theorem~\ref{th:7colors} shows that some given $c_i$ with 6 colors can force $c_{i+1}$ to use 7 colors.
    
    This leads to the question of which problems can be solved from any given coloring, and which ones need to maintain some property over the coloring of each snapshot?
    
    For example, if each snapshot is a tree and we are given a 9-coloring $c_1$ of $G_1\cup G_2$, can we get a temporal coloring and some $I$ such that, for each $i\ge I$, $c_i$ uses at most 8 colors? 
    
    \item Can we find notions of compatibility for other graph problems, such as Maximal Matching and Maximal Independent Set?
\end{itemize}

%new definitions for Coloring and Recoloring Problems in temporal graphs and gave bounds on the number of colors needed in general and some restrictive cases. 

%Several questions are left open. A natural question raised is to improve bounds on the number of colors needed to recolor temporal graphs, e.g., When subgraphs are paths, how to generate "better random paths", what is exactly the number of colors needed (6, 7 or 8). A direct follow up of the above questions is to consider other graph problems, beyond coloring, like maximal matching, maximal independent set, etc. In particular, the setting of temporal graphs opens a perspective of different definitions for this general problem, and of what new complexity measure.

\bibliographystyle{plainurl}
\bibliography{ref}
\end{document}